\documentclass[reqno]{amsart}
\usepackage[margin=1.49in,bottom=1.5in]{geometry}		

\usepackage{amsmath}		
\usepackage{amssymb}		
\usepackage{amsfonts}		
\usepackage{amsthm}		
\usepackage[foot]{amsaddr}		

\usepackage[centercolon=true]{mathtools}		

\mathtoolsset{%
}

\usepackage[utf8]{inputenc}		
\usepackage[T1]{fontenc}		

\usepackage[
cal=cm,
]
{mathalfa}

\usepackage{dsfont}		

\usepackage[sf,mono=false]{libertine}

\usepackage{acronym}		
\renewcommand*{\aclabelfont}[1]{\acsfont{#1}}		
\newcommand{\acli}[1]{\textit{\acl{#1}}}		
\newcommand{\acdef}[1]{\textit{\acl{#1}} \textup{(\acs{#1})}\acused{#1}}		

\usepackage[labelfont={bf,small},labelsep=colon,font=small]{caption}	
\usepackage{subcaption}		
\usepackage[dvipsnames,svgnames]{xcolor}		
\colorlet{MyRed}{FireBrick!50!Crimson}
\colorlet{MyBlue}{DodgerBlue!75!black}
\colorlet{MyGreen}{DarkGreen!85!black}
\colorlet{MyViolet}{DarkMagenta}

\colorlet{MyLightBlue}{DodgerBlue!20}
\colorlet{MyLightGreen}{MyGreen!20}

\colorlet{PrimalColor}{MyBlue}
\colorlet{PrimalFill}{MyLightBlue}
\colorlet{DualColor}{MyRed}

\colorlet{RevColor}{MyRed}
\colorlet{LinkColor}{MediumBlue}

\newcommand{\afterhead}{.\;}		
\newcommand{\para}[1]{\medskip\paragraph{\textbf{#1\afterhead}}}

\usepackage{cancel}		
\usepackage{latexsym}		

\usepackage{pifont}		

\usepackage{tikz}		
\usepackage{tikz-cd}		
\usetikzlibrary{calc,patterns}		
\usepackage{pgfplots}

\usepackage{array}		
\usepackage{booktabs}		
\usepackage[inline,shortlabels]{enumitem}		
\setenumerate{itemsep=\smallskipamount,topsep=\smallskipamount,left=\parindent/2}
\setitemize{itemsep=\smallskipamount,topsep=\smallskipamount,left=\parindent/2}
\setlist[enumerate,1]{label = \upshape(\arabic*)}

\usepackage[kerning=true]{microtype}		

\usepackage{tabto}		
\usepackage{xspace}		

\usepackage[sort&compress]{natbib}		

\bibpunct[, ]{[}{]}{,}{n}{,}{,}

\usepackage{hyperref}
\hypersetup{
final,
colorlinks=true,
linktocpage=true,
pdfstartview=FitH,
breaklinks=true,
pdfpagemode=UseNone,
pageanchor=true,
pdfpagemode=UseOutlines,
plainpages=false,
bookmarksnumbered,
bookmarksopen=false,
bookmarksopenlevel=1,
hypertexnames=true,
pdfhighlight=/O,
urlcolor=LinkColor,linkcolor=LinkColor,citecolor=LinkColor,	
pdftitle={},
pdfauthor={},
pdfsubject={},
pdfkeywords={},
pdfcreator={pdfLaTeX},
pdfproducer={LaTeX with hyperref}
}

\newcommand{\EMAIL}[1]{\email{\href{mailto:#1}{#1}}}

\usepackage[sort&compress,capitalize,nameinlink]{cleveref}		
\crefname{algo}{Algorithm}{Algorithms}
\crefname{assumption}{Assumption}{Assumptions}
\crefname{case}{Case}{Cases}

\usepackage{algorithm}		
\usepackage{algpseudocode}		

\usepackage{thmtools}		
\usepackage{thm-restate}		

\theoremstyle{plain}
\newtheorem{lemma}{Lemma}		
\newtheorem{proposition}{Proposition}		

\newtheorem*{corollary*}{Corollary}		

\theoremstyle{definition}
\newtheorem{definition}{Definition}		
\newtheorem{example}{Example}		

\newtheorem*{definition*}{Definition}		
\newtheorem*{assumption*}{Assumptions}		
\newtheorem*{example*}{Example}		

\theoremstyle{remark}

\newtheorem*{remark*}{Remark}		

\def\endenv{\hfill\mbox{\small$\lozenge$}}

\newcounter{proofpart}

\numberwithin{example}{section}		

\usepackage[showdeletions]{color-edits}		
\newcommand{\debug}[1]{#1}		

\newcommand{\newmacro}[2]{\newcommand{#1}{\debug{#2}}}		
\newcommand{\newop}[2]{\DeclareMathOperator{#1}{\debug{#2}}}		

\DeclarePairedDelimiter{\braces}{\{}{\}}		
\DeclarePairedDelimiter{\bracks}{[}{]}		
\DeclarePairedDelimiter{\parens}{(}{)}		

\DeclarePairedDelimiter{\abs}{\lvert}{\rvert}		

\DeclarePairedDelimiter{\setof}{\{}{\}}		
\DeclarePairedDelimiterX{\setdef}[2]{\{}{\}}{#1:#2}		
\DeclarePairedDelimiterXPP{\exclude}[1]{\mathopen{}\setminus}{\{}{\}}{}{#1}		

\newcommand{\N}{\mathbb{N}}		
\newcommand{\R}{\mathbb{R}}		
\newcommand{\C}{\mathbb{C}}		

\DeclareMathOperator*{\argmax}{arg\,max}		
\DeclareMathOperator*{\argmin}{arg\,min}		

\DeclareMathOperator{\bigoh}{\mathcal{O}}		
\DeclareMathOperator{\cl}{cl}		
\DeclareMathOperator{\diag}{diag}		
\DeclareMathOperator{\dist}{dist}		
\DeclareMathOperator{\dom}{dom}		
\DeclareMathOperator{\eig}{eig}		
\DeclareMathOperator{\im}{im}		
\DeclareMathOperator{\one}{\mathds{1}}		
\DeclareMathOperator{\rank}{rank}		
\DeclareMathOperator{\relint}{ri}		
\DeclareMathOperator{\supp}{supp}		
\DeclareMathOperator{\vol}{vol}		

\newcommand{\cf}{cf.\xspace}		
\newcommand{\eg}{e.g.,\xspace}		
\newcommand{\ie}{i.e.,\xspace}		
\newcommand{\viz}{viz.\xspace}		

\newcommand{\textpar}[1]{\textup(#1\textup)}		

\newcommand{\txs}{\textstyle}		

\newcommand{\alt}[1]{#1'}		
\newcommand{\altalt}[1]{#1''}		

\newmacro{\dd}{\:d}		
\newcommand{\ddt}{\frac{d}{dt}}		
\newcommand{\eps}{\varepsilon}		

\newcommand{\insum}{\sum\nolimits}		

\newmacro{\const}{c}		
\newmacro{\Const}{\rho}		
\newmacro{\coefalt}{\mu}		
\usepackage{xparse}
\NewDocumentCommand{\coef}{O{\lambda}}{\debug{#1}}
\newmacro{\param}{\theta}		
\newmacro{\params}{\Theta}		

\newmacro{\pexp}{p}		
\newmacro{\qexp}{q}		
\newmacro{\rexp}{r}		

\newmacro{\radius}{r}

\newmacro{\beforestart}{0}		
\newmacro{\start}{1}		
\newmacro{\afterstart}{2}		
\newmacro{\running}{\start,\afterstart,\dotsc}		
\newmacro{\halfrunning}{1,3/2,2\dotsc}		

\newmacro{\run}{n}		
\newmacro{\runalt}{k}		
\newmacro{\runaltalt}{\tau}		
\newmacro{\nRuns}{T}		
\newmacro{\runs}{\mathcal{\nRuns}}		

\newmacro{\state}{\denmat}		
\newmacro{\statealt}{\mathbf{Z}}		

\newcommand{\curr}[1][\state]{\debug{#1}_{\run}}		

\newmacro{\tstart}{0}		
\renewcommand{\time}{\debug{t}}		
\newmacro{\timealt}{s}		
\newmacro{\horizon}{T}		

\newmacro{\traj}{x}		
\newmacro{\trajalt}{y}		
\newmacro{\trajaltalt}{z}		

\newmacro{\flow}{\phi}		
\DeclarePairedDelimiterXPP{\flowof}[2]{\flow_{#1}}{(}{)}{}{#2}		

\newop{\Nash}{NE}		
\newop{\CE}{CE}		
\newop{\CCE}{CCE}		
\newop{\NI}{NI}		

\newop{\brep}{br}		
\newop{\reg}{Reg}		
\newop{\preg}{\overline{Reg}}		
\newop{\val}{val}		

\newmacro{\strat}{x}		
\newmacro{\stratalt}{\alt\strat}		
\newmacro{\strats}{\mathcal{X}}		

\newmacro{\play}{i}		
\newmacro{\playalt}{j}		
\newmacro{\playaltlalt}{k}		
\newmacro{\nPlayers}{N}		
\newmacro{\players}{\mathcal{\nPlayers}}		

\newmacro{\pure}{\alpha}		
\newmacro{\puredummy}{\kappa}		
\newmacro{\purealt}{\beta}		
\newmacro{\purealtalt}{\gamma}		
\newmacro{\nPures}{A}		
\newmacro{\pures}{\mathcal{\nPures}}		

\newmacro{\loss}{\ell}		
\newmacro{\pay}{u}		
\newmacro{\payv}{v}		
\newmacro{\payfield}{\mathbf{V}}		
\newmacro{\pot}{\Phi}		

\newmacro{\game}{\mathcal{G}}		
\newmacro{\gamefull}{\game(\players,\points,\pay)}		

\newmacro{\fingame}{\Gamma}		
\newmacro{\fingamefull}{\Gamma(\players,\pures,\pay)}		

\newmacro{\gmat}{g}		
\newmacro{\gdist}{\dist_{\gmat}}
\newmacro{\mfld}{M}		
\newmacro{\form}{\omega}		

\newmacro{\tvec}{z}		
\newmacro{\uvec}{u}		

\newmacro{\ball}{\basin}		
\newmacro{\sphere}{\mathbb{S}}		

\newmacro{\graph}{\mathcal{G}}
\newmacro{\vertices}{\mathcal{V}}
\newmacro{\edges}{\mathcal{E}}

\newcommand{\herm}[1][]{\mathbb{H}^{#1}}		
\newcommand{\psd}[1][]{\mathbb{H}_{+}^{#1}}		

\newmacro{\mat}{\mathbf{A}}		
\newmacro{\matval}{\lambda}		
\newmacro{\matvals}{\mathbf{\Lambda}}		
\newmacro{\matvec}{\mathbf{u}}		
\newmacro{\matvecs}{\mathbf{U}}		

\newmacro{\matalt}{\mathbf{M}}		
\newmacro{\hmat}{\mathbf{H}}		

\newmacro{\diagmat}{\mathbf{\Lambda}}		
\newmacro{\unitmat}{\mathbf{U}}		

\newop{\row}{row}		
\newop{\col}{col}		

\newmacro{\ones}{\mathbf{1}}		
\newmacro{\eye}{\mathbf{I}}		
\newmacro{\zer}{\mathbf{0}}		

\DeclareMathOperator{\tr}{tr}		
\DeclarePairedDelimiterXPP{\trof}[1]{\tr}{[}{]}{}{#1}

\DeclarePairedDelimiter{\norm}{\lVert}{\rVert}		
\DeclarePairedDelimiterXPP{\dnorm}[1]{}{\lVert}{\rVert}{_{\ast}}{#1}		

\DeclarePairedDelimiterXPP{\onenorm}[1]{}{\lVert}{\rVert}{_{1}}{#1}		
\DeclarePairedDelimiterXPP{\twonorm}[1]{}{\lVert}{\rVert}{_{2}}{#1}		
\DeclarePairedDelimiterXPP{\supnorm}[1]{}{\lVert}{\rVert}{_{\infty}}{#1}		
\DeclarePairedDelimiterXPP{\frobnorm}[1]{}{\lVert}{\rVert}{_{F}}{#1}		

\DeclarePairedDelimiter{\bra}{\langle}{\rvert}		
\DeclarePairedDelimiter{\ket}{\lvert}{\rangle}		
\DeclarePairedDelimiterX{\braket}[2]{\langle}{\rangle}{#1\mathopen{}\delimsize\vert\mathopen{}#2}
\DeclarePairedDelimiter{\obs}{\langle}{\rangle}		

\newmacro{\vecspace}{\mathcal{V}}		
\newmacro{\subspace}{\mathcal{W}}		

\newmacro{\coord}{\pure}		
\newmacro{\coordalt}{\purealt}		
\newmacro{\coordaltalt}{\purealtalt}		
\newmacro{\nCoords}{d}		
\newmacro{\dims}{\nCoords}		
\newmacro{\vdim}{\nCoords}		

\newmacro{\pvec}{z}		
\newmacro{\pvecalt}{r}		
\newmacro{\bvec}{e}		
\newmacro{\bvecs}{\mathcal{E}}		
\newmacro{\cvec}{b}     
\newmacro{\cvecalt}{d}     

\newmacro{\pspace}{\vecspace}		
\newmacro{\dspace}{\vecspace^{\ast}}		

\newmacro{\dvec}{\dpoint}		
\newmacro{\dbvec}{\eps}		

\newmacro{\dpoint}{\mathbf{Y}}		
\newmacro{\dpointalt}{\alt\dpoint}		
\newmacro{\dpointaltalt}{\altalt\dpoint}		
\newmacro{\dpoints}{\mathcal{Y}}		

\newmacro{\dstate}{\dmat}		
\newmacro{\dbase}{v}		

\newcommand{\defeq}{\coloneqq}		

\newcommand{\from}{\colon}		

\newmacro{\argdot}{\textrm{\small\textbullet}}		
\newmacro{\fun}{f}		

\newop{\Opt}{Opt}		
\newop{\Sol}{Sol}		
\newop{\gap}{Gap}		
\newop{\orcl}{Or}		

\newmacro{\tfun}{f}		
\newmacro{\obj}{f}		
\newmacro{\objalt}{g}		
\newmacro{\sobj}{F}		

\newmacro{\gvec}{g}		
\newmacro{\oper}{A}		
\newmacro{\vecfield}{v}		

\newcommand{\sol}[1][\point]{#1^{\ast}}		
\newmacro{\solvec}{\vecfield(\sol)}		
\newmacro{\solpay}{\eq[\payfield]}		

\newmacro{\test}{\mathbf{P}}

\newmacro{\signal}{V}		
\newmacro{\step}{\gamma}		
\newmacro{\learn}{\eta}		

\newmacro{\vbound}{G}		
\newmacro{\lips}{L}		
\newmacro{\strong}{\mu}		
\newmacro{\smooth}{\beta}		

\newop{\tspace}{T}		
\newop{\tcone}{TC}		
\newop{\dcone}{\tcone^{\ast}}		
\newop{\ncone}{NC}		
\newop{\pcone}{PC}		
\newop{\hull}{\Delta}		

\newmacro{\cvx}{\mathcal{C}}		
\newmacro{\subd}{\partial}		

\newmacro{\minmax}{\mathcal{L}}		

\newmacro{\minvar}{\point_{1}}		
\newmacro{\minvaralt}{\alt\minvar}		
\newmacro{\minvars}{\points_{1}}		
\newmacro{\minsol}{\sol[\minvar]}		

\newmacro{\maxvar}{\point_{2}}		
\newmacro{\maxvaaltr}{\alt\maxvar}		
\newmacro{\maxvars}{\points_{2}}		
\newmacro{\maxsol}{\sol[\maxvar]}		

\newop{\Eucl}{\mathbf{\Pi}}		
\newop{\logit}{\mathbf{\Lambda}}		
\newop{\dkl}{KL}		

\newmacro{\hreg}{h}		
\newmacro{\hconj}{\hreg^{\ast}}		
\newmacro{\breg}{D}		
\newmacro{\mprox}{P}		
\newmacro{\mirror}{\mathbf{Q}}		
\newmacro{\fench}{F}		
\newmacro{\depth}{H}		
\newmacro{\hstr}{K}		
\newmacro{\hker}{\theta}		

\newmacro{\proxdom}{\denmats_{\hreg}}		
\newmacro{\proxdomi}{\denmats_{\hreg_{\play}}}		
\newmacro{\zone}{\mathbb{D}}		

\DeclarePairedDelimiterXPP{\proxof}[2]{\mprox_{#1}}{(}{)}{}{#2}		

\newmacro{\point}{z}		
\newmacro{\pointalt}{\alt\point}		
\newmacro{\pointaltalt}{\altalt\point}		
\newmacro{\points}{\mathcal{Z}}		
\newmacro{\intpoints}{\relint\points}		

\newmacro{\base}{p}		
\newmacro{\basealt}{q}		
\newmacro{\basealtalt}{u}		

\newmacro{\open}{\mathcal{U}}		
\newmacro{\closed}{\mathcal{C}}		
\newmacro{\cpt}{\mathcal{K}}		
\newmacro{\nhd}{\mathcal{U}}		
\newmacro{\nhdalt}{\alt\nhd}		

\newop{\ex}{\mathbb{E}}		
\newop{\prob}{\mathbb{P}}		
\newop{\Var}{Var}		
\newop{\simplex}{\hull}		

\DeclarePairedDelimiterXPP{\exof}[1]{\ex}{[}{]}{}{
 #1}

\DeclarePairedDelimiterXPP{\probof}[1]{\prob}{(}{)}{}{
 #1}

\DeclarePairedDelimiterXPP{\oneof}[1]{\one}{\{}{\}}{}{
 #1}

\newmacro{\sample}{\omega}		
\newmacro{\samples}{\Omega}		

\newmacro{\seed}{\theta}		
\newmacro{\seeds}{\Theta}		

\newmacro{\filter}{\mathcal{F}}		
\newmacro{\probspace}{(\samples,\filter,\prob)}		

\newmacro{\history}{\mathcal{H}}		

\newmacro{\event}{E}       
\newmacro{\eventalt}{H}       

\newmacro{\mean}{\mu}		
\newmacro{\sdev}{\sigma}		
\newmacro{\variance}{\sdev^{2}}		

\newmacro{\proper}{\tau}		

\newmacro{\error}{Z}		
\newmacro{\noise}{U}		
\newmacro{\bias}{b}		
\newmacro{\brown}{W}		

\newmacro{\serror}{\theta}		
\newmacro{\snoise}{\xi}		
\newmacro{\sbias}{\psi}		

\newmacro{\sbound}{M}		
\newmacro{\bbound}{B}		
\newmacro{\noisepar}{\sdev}		
\newmacro{\noisevar}{\variance}		

\addauthor[Kyriakos]{KL}{magenta}

\newmacro{\evector}{\mathbf{u}}
\newmacro{\evalue}{x}

\newmacro{\lag}{\mathcal{L}}		
\newmacro{\hit}{\tau}		
\newmacro{\cont}{\mathcal{D}}		
\newmacro{\limitden}{\denmat_{\infty}}		
\newmacro{\limitnhd}{\strats^{*}}		
\newmacro{\bench}{\mathbf{P}}		
\newmacro{\invar}{\mu}		
\newmacro{\quotientmap}{q}		
\newmacro{\quotientspace}{\dmats_0}		
\newmacro{\quotientmirror}{\mirror_0}		
\newmacro{\Leb}{\lambda}		
\newmacro{\rankeq}{r}		
\newmacro{\limitnhdrank}{\limitnhd_\limitrank}		
\newmacro{\basin}{\nhd_0}		
\newmacro{\denmateq}{\denmat^*}		

\addauthor[Pan]{PM}{MediumBlue}

\newmacro{\hilbert}{\mathcal{H}}		
\newmacro{\quant}{\psi}		
\newmacro{\quantalt}{\alt\psi}		
\newmacro{\quants}{\Psi}		

\newmacro{\povm}{\mathbf{P}}		
\newmacro{\pdist}{P}		
\newmacro{\outcome}{\sample}		
\newmacro{\outcomes}{\samples}		
\newmacro{\payobs}{U}		

\newmacro{\denop}{\chi}		
\newmacro{\denmat}{\mathbf{X}}		
\newmacro{\denval}{x}		
\newmacro{\denvals}{\diag(\mathbf{x})}		
\newmacro{\denvec}{\mathbf{u}}		
\newmacro{\denvecs}{\mathbf{U}}		

\newmacro{\denmatalt}{\alt\denmat}		
\newmacro{\denmats}{\boldsymbol{\mathcal{X}}}		
\newmacro{\basemat}{\mathbf{P}}		

\newmacro{\dmat}{\mathbf{Y}}		
\newmacro{\dmatalt}{\alt\dmat}		
\newmacro{\dmats}{\boldsymbol{\mathcal{Y}}}		

\newmacro{\pflow}{\boldsymbol{\chi}}		
\DeclarePairedDelimiterXPP{\pflowof}[2]{\pflow_{#1}}{(}{)}{}{#2}		

\newmacro{\dflow}{\boldsymbol{\psi}}		
\DeclarePairedDelimiterXPP{\dflowof}[2]{\dflow_{#1}}{(}{)}{}{#2}		

\newmacro{\pnhd}{\boldsymbol{\mathcal{U}}}		
\newmacro{\pnhdalt}{\alt\pnhd}		

\newmacro{\dnhd}{\boldsymbol{\mathcal{W}}}		

\newmacro{\limpoint}{\hat\denmat}		

\newmacro{\payent}{V}		
\newmacro{\denent}{X}		

\newcommand{\eq}{\sol[\denmat]}		
\newcommand{\eqs}{\sol[\denmats]}		

\newmacro{\qgame}{\mathcal{Q}}		
\newmacro{\qgamefull}{\qgame(\players,\quants,\pay)}		
\newmacro{\qgamecont}{\qgame(\players,\denmats,\pay)}		

\newcommand{\stateof}[2][]{\denmat_{#1}\parens{#2}}
\newcommand{\dstateof}[2][]{\dmat_{#1}\parens{#2}}

\newmacro{\energy}{E}		

\newmacro{\A}{\boldsymbol{\mathcal{A}}}		

\addauthor[Nick]{NB}{DarkMagenta}

\begin{document}


\title{Learning in Quantum Games}

\author
[K.~Lotidis]
{Kyriakos Lotidis$^{\ast}$}
\address{$^{\ast}$\,%
Department of Management Science \& Engineering, Stanford University.}
\EMAIL{klotidis@stanford.edu}
\author
[P.~Mertikopoulos]
{Panayotis Mertikopoulos$^{\sharp}$}
\address{$^{\sharp}$\,%
Univ. Grenoble Alpes, CNRS, Inria, Grenoble INP, LIG, 38000 Grenoble, France.}
\EMAIL{panayotis.mertikopoulos@imag.fr}
\author
[N.~Bambos]
{Nicholas Bambos$^{\ddag}$}
\address{$^{\ddag}$\,%
Department of Electrical Engineering and Management Science \& Engineering, Stanford University.}
\EMAIL{bambos@stanford.edu}



\newacro{LHS}{left-hand side}
\newacro{RHS}{right-hand side}
\newacro{iid}[i.i.d.]{independent and identically distributed}

\newacro{NE}{Nash equilibrium}
\newacroplural{NE}[NE]{Nash equilibria}
\newacro{VI}{variational inequality}
\newacroplural{VI}[VIs]{variational inequalities}

\newacro{DGF}{distance-generating function}
\newacro{KKT}{Karush\textendash Kuhn\textendash Tucker}
\newacro{ODE}{ordinary differential equation}

\newacro{PVM}{projection-valued measure}
\newacro{POVM}{positive operator-valued measure}
\newacro{FTRL}{``follow the regularized leader''}
\newacro{FTQL}{``follow the quantum regularized leader''}
\newacro{q-FTRL}{``quantum follow the regularized leader''}
\newacro{QRD}{quantum replicator dynamics}
\newacro{GAN}{generative adversarial network}
\newacro{QGAN}{quantum generative adversarial network}

\newacro{MXW}{matrix exponential weights}
\newacro{MMW}{matrix multiplicative\,/\,exponential weights}
\newacro{MRD}{matrix replicator dynamics}

\begin{abstract}
%
%
In this paper, we introduce a class of learning dynamics for general quantum games, that we call \acdef{FTQL}, in reference to the classical \acs{FTRL} template for finite games.
We show that the induced quantum state dynamics decompose into
\begin{enumerate*}
[(\itshape i\hspace*{1pt}\upshape)]
\item
a classical, commutative component which governs the dynamics of the system's eigenvalues in a way analogous to the evolution of mixed strategies under \acs{FTRL};
and
\item
a non-commutative component for the system's eigenvectors which has no classical counterpart.
\end{enumerate*}
Despite the complications that this non-classical component entails, we find that the \ac{FTQL} dynamics incur no more than constant regret in \emph{all} quantum games.
Moreover, adjusting classical notions of stability to account for the nonlinear geometry of the state space of quantum games,
we show that only \emph{pure quantum equilibria} can be stable and attracting under \ac{FTQL}
while, as a partial converse,
pure equilibria that satisfy a certain ``variational stability'' condition are always attracting.
Finally, we show that the \ac{FTQL} dynamics are \emph{Poincaré recurrent} in quantum min-max games, extending in this way a very recent result for the \acl{QRD}.
\vspace{-\baselineskip}
\end{abstract}
\maketitle

\acresetall		

\section{Introduction}
\label{sec:introduction}

The advent of quantum information theory \textendash\ and, with it, the associated ``quantum advantage'' \cite{Pre18,AABB+19,ZWDC+20} \textendash\ has had a profound impact on computer science and machine learning, from quantum cryptography and shadow tomography \cite{Aar20}, to \acp{QGAN} and adversarial learning \cite{DK18,CHL19,LW18}.
At a high level, the advantages of quantum-based computing are owed to the possibility of preparing superpositions of binary-state quantum systems known as \emph{qubits}:
classical bits cannot lie in superposition, so the calculations that can be performed by classical computers are \emph{de facto} limited by their binary alphabet and memory structure.
In light of this, quantum computing has the potential to greatly accelerate the development of artificial intelligence algorithms and models, with Google's ``Sycamore'' $54$-qubit processor training an autonomous vehicle model in less than $200$ seconds \cite{AABB+19}.

In a similar manner, when such models are deployed in a multi-agent context \textendash\ \eg as in the case of \acp{QGAN} or autonomous vehicles \textendash\ the landscape changes drastically relative to classical non-cooperative frameworks.
The main reason for this is again the ``quantum advantage'':
due to the intricacies of decoherence and entaglement \textendash\ two quantum notions that have no classical counterpart \textendash\ quantum players can have a distinct advantage over ``classical'' players, achieving higher payoffs at equilibrium than would otherwise be possible \cite{Mey99,EWL99}.
This is again owed to the fact that probabilistic mixing works differently in the quantum and classical worlds:
in classical games, a mixed \emph{strategy} is a probabilistic convex combination of the constituent pure strategies;
in quantum games, a mixed \emph{state} is a probabilistic mixture of the quantum \emph{projectors} associated to each constituent state.
Because of this, a mixed quantum state can return payoffs that lie outside the convex hull of classical mixed strategies, thus providing a tangible advantage to players with access to quantum technologies \textendash\ \eg the ability to encode their action in a qubit register, which is then submitted to a ``referee'' (the natural mechanism determining the payoffs of a quantum game).

Of course, the extent to which the advantage of quantum players manifests itself is contingent on the players' actually reaching an equilibrium.
The recent work of \citet{BW22} has shown that the problem of computing an approximate \acl{NE} of a quantum game is included in PPAD, so, by the seminal work of \citet{DGP06,DGP09-acm}, it must be complete for this class (since computing a quantum equilibrium is at least as hard as computing a classical one).
Thus, given that the dimensionality of a quantum game is exponential in the number of qubits available to each player, computing a \acl{NE} of a quantum game quickly becomes an intractable affair, in all but the smallest games.
On that account, it seems more reasonable to turn to an online learning paradigm where each player seeks to minimize their individual regret, and instead ask:
\begin{center}
\itshape
Are all equilibrium outcomes equally likely under a quantum no-regret learning scheme?
\\
Is there a class of equilibria with an inherent selection bias \textendash\ either for or against?
\end{center}

\para{Our contributions}

First, to achieve no-regret in a quantum setting, we introduce a flexible model for learning in general $\nPlayers$-player quantum games based on the popular \acf{FTRL} template for \emph{finite} games \cite{SSS06,SS11}.
The resulting model, which we call \acdef{FTQL}, contains as a special case the \ac{MMW} dynamics that have been used extensively in quantum games and matrix learning \cite{TRW05,JW09,KSST12,ACHK+18,JPS22}, and which give rise to the \acl{QRD} \cite{Hid06,JPS22}.
Importantly, as we show in \cref{sec:dynamics}, the mixed-state dynamics of \ac{FTQL} decompose into a ``classical'' part (eigenvalues evolve as the \ac{FTRL} dynamics in finite games), plus a ``quantum'' component capturing the evolution of the system's eigenfunctions (and which has no classical analogue).

In terms of regret minimization, all \ac{FTQL} dynamics incur at most constant regret, so they represent a compelling choice from a learning standpoint.
However, deriving the dynamics' equilibrium convergence properties is significantly more difficult because of the nonlinear geometry of the game's state space.
Specifically, in contrast to finite games (where pure strategies are isolated extreme points), the pure states of a quantum game form a continuous manifold of stationary points (all of them extreme), so the study of stability and convergence questions becomes a highly involved affair.
Nonetheless, despite these topological complications, we show that \ac{FTQL} enjoys the following fundamental properties:
\begin{enumerate*}
[(\itshape i\hspace*{1pt}\upshape)]
\item
\aclp{NE} are stationary;
\item
limits of interior trajectories and Lyapunov stable states are Nash;
\item
only \emph{pure quantum equilibria} can be stable and attracting under \ac{FTQL} (up to the exclusion of trivial stationary states);
and
\item
as a partial converse, we show that pure states that satisfy a certain ``variational stability'' condition are attracting, irrespective of the chosen regularizer.
\end{enumerate*}
On that account, our results lead to an implicit quantum ``purification'' principle:
under \ac{FTQL},
\emph{mixed states are inherently fragile, and only pure quantum states can be consistently attracting.}%
\footnote{In the classical world, a version of the above collection of results is sometimes referred to as the ``folk theorem'' of evolutionary game theory \cite{Cre03,HS03,MS16,FVGL+20}.}

Finally, we complement our results with a closer look at two-player, zero-sum quantum games, where we show that
\ac{FTQL} exhibits a cycling property known as \emph{Poincaré recurrence:}
almost all trajectories of play return arbitrarily close to their starting point infinitely often.
This result is a broad generalization of a recent result by \citet{JPS22}, who established this property for the \ac{MMW} dynamics.
In this regard, our result shows that the \ac{MMW} result is not a coincidence:
despite the ``quantum advantage'', perfect competition cannot be resolved by the dynamics of regularized learning.

To simplify the presentation, we focus throughout on models that evolve in continuous time.
This allows us to sidestep issues having to do with hyperparameter tuning and the like, and instead spotlight the essential aspects of the theory.

\section{Preliminaries}
\label{sec:prelims}

We start by briefly reviewing some basics of quantum game theory and introducing the necessary context for our results.

\para{Notation}
Given a (complex) Hilbert space $\hilbert$, we will use Dirac's bra-ket notation to distinguish between an element $\ket{\quant}$ of $\hilbert$ and its adjoint $\bra{\quant}$;
otherwise, when a basis is implied by the context, we will use the dagger notation ``$\dag$'' to denote the Hermitian transpose $\quant^{\dag}$ of $\quant$.
We will also write $\herm[\vdim]$ for the space of $\vdim\times\vdim$ Hermitian matrices, and $\psd[\vdim]$ for the cone of positive-semidefinite matrices in $\herm[\vdim]$.
Finally, given a real function $\fun\from\R\to\R$ and a Hermitian matrix $\denmat\in\herm[\vdim]$ with unitary eigen-decomposition $\denmat = \sum_{\coord=1}^{\vdim} \denval_{\coord} \denvec_{\coord} \denvec_{\coord}^{\dag}$,
we will write $\fun(\denmat)$ for the (likewise Hermitian) matrix $\fun(\denmat) = \sum_{\coord=1}^{\vdim} \fun(\denval_{\coord}) \denvec_{\coord}\denvec_{\coord}^{\dag}$.

\para{Quantum games}

Following \cite{EWL99,GW07}, a \emph{quantum game} consists of the following primitives:
\begin{enumerate}
\item
A finite set of \emph{players} $\play\in\players = \{1,\dotsc,\nPlayers\}$.
\item
Each player $\play\in\players$ has access to a complex Hilbert space $\hilbert_{\play} \cong \C^{\vdim_{\play}}$ describing the set of (pure) \emph{quantum states} available to the player (typically a discrete register of qubits).
In more detail, a quantum state is an element $\quant_{\play}$ of $\hilbert_{\play}$ with unit norm, so the set of all such states is the unit sphere $\quants_{\play} \defeq \setdef{\quant_{\play}\in\hilbert_{\play}}{\norm{\quant_{\play}}=1}$ of $\hilbert_{\play}$.
We will also write $\quants \defeq \prod_{\play} \quants_{\play}$ for the space of all ensembles $\quant = (\quant_{1},\dotsc,\quant_{\nPlayers})$ of pure states $\quant_{\play}\in\quants_{\play}$ that are independently prepared by each player $\play\in\players$.
\item
The players' rewards are determined by their individual \emph{payoff functions} $\pay_{\play}\from\quants \to \R$.
These payoff functions are not arbitrary, but are obtained from a joint \acdef{POVM} quantum measurement process that unfolds as follows \cite{CN10}:
First, we assume given a finite set of possible \emph{measurement outcomes} $\outcome \in \outcomes$ that a referee can observe from the players' quantum states (\eg measure a player-prepared qubit to be ``up'' or ``down'').
Each such outcome $\outcome\in\outcomes$ is associated to a positive semi-definite operator $\povm_{\outcome}\from\hilbert\to\hilbert$ that acts on the tensor product $\hilbert \defeq \bigotimes_{\play} \hilbert_{\play}$ of the players' individual state spaces;
we further assume that $\sum_{\outcome\in\outcomes} \povm_{\outcome} = \eye$ so the joint probability of observing $\outcome\in\outcomes$ when the system is at state $\quant\in\quants$ is
\begin{equation}
\pdist_{\outcome}(\quant)
	= \bra{\quant_{1}\otimes\dotsm\otimes\quant_{\nPlayers}}
		\povm_{\outcome}
		\ket{\quant_{1}\otimes\dotsm\otimes\quant_{\nPlayers}}
\end{equation}
The payoff to each player $\play\in\players$ is given by the outcome of this measurement process via a \emph{payoff observable} $\payobs_{\play}\from\outcomes\to\R$;
specifically, in this context, $\pay_{\play}(\quant)$ denotes the player's expected payoff at state $\quant\in\quants$, \viz
\begin{equation}
\label{eq:pay-pure}
\pay_{\play}(\quant)
	\defeq \obs{\payobs_{\play}}
	\equiv \insum_{\outcome} \pdist_{\outcome}(\quant) \, \payobs_{\play}(\outcome).
\end{equation}
\end{enumerate}
A \emph{quantum game}
is then defined as a tuple $\qgame \equiv \qgamefull$ with players, quantum states, and payoff functions as above.

\para{Mixed states}

In addition to pure states, each player $\play\in\players$ can also prepare probabilistic mixtures thereof, known as \emph{mixed states}.
In contrast to mixed strategies in classical, finite games, these mixed states are \emph{not} convex combinations of their pure counterparts;
instead, given a family of pure quantum states $\quant_{\play\pure_{\play}} \in \quants_{\play}$ indexed by $\pure_{\play}\in\pures_{\play}$, a mixed state is described by a \emph{density matrix} of the form
\begin{equation}
\label{eq:denmat}
\denmat_{\play}
	= \smashoperator{\sum_{\pure_{\play}\in\pures_{\play}}}
		\strat_{\play\pure_{\play}}
		\ket{\quant_{\play\pure_{\play}}}
		\bra{\quant_{\play\pure_{\play}}}
\end{equation}
where
$\strat_{\play\pure_{\play}} \geq 0$ is the mixing weight of $\quant_{\play\pure_{\play}}$,
and
we assume that $\tr\denmat_{\play} = 1$ (the states $\quant_{\play\pure_{\play}}$ are not assumed to be orthogonal in this context).
By Born's rule, this means that if each player $\play\in\players$ prepares a mixed state according to $\denmat_{\play}$, the probability of observing $\outcome\in\outcomes$ under $\denmat = (\denmat_{1},\dotsc,\denmat_{\nPlayers})$ will be
\(
\pdist_{\outcome}(\denmat)
	= \sum_{\pure} \strat_{\pure} \bra{\quant_{\pure}} \povm_{\outcome} \ket{\quant_{\pure}},
\)
where, in multi-index notation,
$\pure = (\pure_{1},\dotsc,\pure_{\nPlayers})$,
$\strat_{\pure} = \prod_{\play} \strat_{\play\pure_{\play}}$,
and
$\quant_{\pure} = \bigotimes_{\play} \quant_{\play\pure_{\play}}$.
Thus, in a slight abuse of notation, the expected payoff to player $\play\in\players$ under $\denmat$ will be
\begin{align}
\label{eq:pay-mix}
\pay_{\play}(\denmat)
	&= \sum_{\outcome\in\outcomes} \sum_{\pure\in\pures}
		\strat_{\pure} \pdist_{\outcome}(\quant_{\pure}) \, \payobs_{\play}(\outcome)
	= \sum_{\pure\in\pures} \strat_{\pure}
		\pay_{\play}(\quant_{\pure})
	\notag\\
	&= \smashoperator{\sum_{\pure_{1}\in\pures_{1}}} \:\dotsi\: \smashoperator{\sum_{\pure_{\nPlayers}\in\pures_{\nPlayers}}}
		\strat_{\pure_{1}} \!\dotsm \strat_{\pure_{\nPlayers}}\,
		\pay_{\play}(\quant_{\pure_{1}},\dotsc,\quant_{\pure_{\nPlayers}}).
\end{align}

\para{Contrasting to other classes of games}

The expression \eqref{eq:pay-mix} for a player's expected payoff under a mixed state is reminiscent of mixed extensions of classical finite games, but this association is very tenuous.
From a conceptual standpoint, the principal differences are as follows:
\begin{enumerate}
\item
There is an infinite continuum of pure states $\quant\in\quants$, not a finite number thereof (as is the case in finite games).
\item
The decomposition \eqref{eq:denmat} of a density matrix into pure states is not unique;
generically, there may be a continuum of (non-equivalent) families of pure states and mixing weights giving rise to the same density matrix.
\item
The convex superposition $\coef\quant + (1-\coef)\quantalt$ of two pure states $\quant$ and $\quantalt$ may give rise to quantum interference terms of the form $\ket{\quant}\bra{\quantalt}$ and $\ket{\quantalt}\bra{\quant}$ in the induced payoff;
these cross-terms have no analogue in finite games.
\end{enumerate}

Because of the above, treating a quantum game as a ``tensorial'' extension of a finite game can be misleading.
Instead, it would be more appropriate to view a quantum game as a \emph{continuous game} where each player $\play\in\players$ controls a matrix variable $\denmat_{\play}$ drawn from the ``spectraplex''
\begin{equation}
\label{eq:denmats}
\denmats_{\play}
	= \setdef{\denmat_{\play}\in\psd[\vdim_{\play}]}{\tr\denmat_{\play} = 1}
\end{equation}
and the player's payoff function $\pay_{\play}\from\denmats \equiv \prod_{\playalt} \denmats_{\playalt} \to \R$ is linear in every player's density matrix $\denmat_{\playalt} \in \denmats_{\playalt}$, $\playalt\in\players$.

\para{\Acl{NE}}

In our quantum setting, the classical solution concept of a \acdef{NE} characterizes mixed quantum states $\eq \in \denmats$ which discourage unilateral deviations in the sense that
\begin{equation}
\label{eq:Nash}
\tag{NE}
\pay_{\play}(\eq)
	\geq \pay_{\play}(\denmat_{\play};\eq_{-\play})
	\quad
	\text{for all $\denmat_{\play}\in\denmats_{\play}$, $\play\in\players$}
\end{equation}
where we write $(\denmat_{\play};\denmat_{-\play}) = (\denmat_{1},\dots,\denmat_{\play},\dotsc,\denmat_{\nPlayers})$ for the choice of player $\play$ relative to all other players.
Since $\denmats_{\play}$ is convex and $\pay_{\play}$ is linear in $\denmat_{\play}$, the existence of \aclp{NE} follows from the seminal theorem of \citet{Deb52}.

Now, letting
\begin{equation}
\label{eq:payv}
\payfield_{\play}(\denmat)
	= \nabla_{\denmat_{\play}^{\top}} \pay_{\play}(\denmat)
\end{equation}
denote the individual payoff gradient of player $\play$, standard arguments \cite{FP03,SFPP10} show that the \aclp{NE} of a quantum game $\qgame$ are precisely the solutions of the \acl{VI}
\begin{equation}
\label{eq:VI}
\tag{VI}
\trof{\payfield(\eq) (\denmat - \eq)}
	\leq 0
	\quad
	\text{for all $\denmat\in\denmats$}
\end{equation}
where $\payfield(\denmat) = (\payfield_{1}(\denmat),\dotsc,\payfield_{\nPlayers}(\denmat))$.
We note here that, since $\pay_{\play}$ is linear in $\denmat_{\play}$, the $\play$-th block $\payfield_{\play}(\denmat)$ of $\payfield(\denmat)$ does not depend on $\denmat_{\play}$, and we have
\begin{equation}
\label{eq:pay-lin}
\pay_{\play}(\denmat_{\play};\denmat_{-\play})
	= \trof{\denmat_{\play} \payfield_{\play}(\denmat)}
	\quad
	\text{for all $\denmat\in\denmats$}.
\end{equation}
We will use these properties freely in the sequel.

\para{Regret}

Complementing the notion of a \acl{NE}, an important rationality requirement in dynamic environments is the minimization of a player's \emph{regret}, \ie the performance gap between the player's expected cumulative payoff over time versus the payoff of the best fixed state in hindsight.
Formally, the regret of the $\play$-th player against the trajectory of play $\stateof{\time} \in \denmats$, $\time\geq0$, is defined as
\begin{equation}
\label{eq:regret}
\reg_{\play}(\horizon)
	= \max_{\denmatalt_{\play}\in\denmats_{\play}}
		\int_{0}^{\horizon} \bracks{\pay_{\play}(\denmatalt_{\play};\stateof[-\play]{\time}) - \pay_{\play}(\stateof{\time})}
		\dd\time
\end{equation}
and we say that player $\play$ has \emph{no regret} if $\reg_{\play}(\horizon) = o(\horizon)$.
In the rest of the paper, we will focus on learning dynamics that incur no regret, and we will examine their convergence properties relative to the game's \aclp{NE}.

\section{Learning dynamics}
\label{sec:dynamics}


\para{Learning via quantum regularization}

In classical, finite games, the most widely studied class of no-regret dynamics is the so-called \acdef{FTRL} family of algorithms \cite{SS11,SSS06,MS16}.
The main idea behind this popular template is the following:
at each instance $\time\geq0$, every player $\play\in\players$ plays a mixed strategy that maximizes the player's cumulative payoff minus a certain regularization penalty.
In this way, strategies that perform consistently better tend to be preferred over their underpeforming counterparts, while the ``regularization penalty'' introduces a certain degree of exploration to avoid getting stuck.

In the quantum regime, the role of mixed strategies is played by the game's mixed quantum states, so the reinforcement mechanism behind \ac{FTRL} leads to the matrix-valued \acdef{FTQL} dynamics
\begin{equation}
\label{eq:FTQL-int}
\stateof[\play]{\time}
	= \argmax\limits_{\denmat_{\play}\in\denmats_{\play}}
		\braces*{
			\int_{0}^{\time} \pay_{\play}(\denmat_{\play};\stateof[-\play]{\timealt}) \dd\timealt
			- \hreg_{\play}(\denmat_{\play})
		}
\end{equation}
where $\hreg_{\play}\from\denmats_{\play}\to\R$ denotes the \emph{penalty function} \textendash\ or \emph{regularizer} \textendash\ of player $\play$ (discussed in detail below).
As stated, the \ac{FTQL} dynamics \eqref{eq:FTQL-int} are in integral form, which is not particularly well-suited for our analysis.
Instead, to obtain a more concrete, autonomous reformulation, consider as a first step the \emph{regularized best response maps}
\begin{equation}
\label{eq:mirror}
\mirror_{\play}(\dmat_{\play})
	\defeq \argmax\nolimits_{\denmat_{\play} \in \denmats_{\play}} \braces{\trof{\dmat_{\play}\denmat_{\play}} - \hreg_{\play}(\denmat_{\play})}
\end{equation}
defined for all Hermitian $\dmat_{\play} \in \dmats_{\play} \defeq \herm[\vdim_{\play}]$.
Then,
in view of \eqref{eq:pay-lin},
the integral dynamics \eqref{eq:FTQL-int} can be recast in differential form as
\begin{equation}
\label{eq:FTQL}
\tag{FTQL}
\dot\dstate_{\play}(\time)
	= \payfield_{\play}(\stateof{\time})
	\qquad
\stateof[\play]{\time}
	= \mirror_{\play}(\dstateof[\play]{\time}).
\end{equation}

The dynamics \eqref{eq:FTQL} will be the basis of our analysis, so some remarks are in order.
First, in terms of interpretation, \eqref{eq:FTQL} can be seen as a gradient-following process coupled with a regularized state selection scheme \textendash\ the dynamics $\dot\dmat_{\play} = \payfield_{\play}$ and the mapping $\mirror_{\play}\from\dmat_{\play} \mapsto \denmat_{\play}$ respectively.
In this regard, the regularizer $\hreg_{\play}$ which underlies the definition of $\mirror_{\play}$ plays a crucial role, and different choices of $\hreg_{\play}$ may yield very different dynamics.
For concreteness, we will only assume in the sequel that $\hreg_{\play}$ is a trace function of the form
\(
\hreg_{\play}(\denmat_{\play})
	= \trof{\hker_{\play}(\denmat_{\play})}
\)
where $\hker_{\play}\from[0,1] \to \R$ is continuous on $[0,1]$ and has $\inf_{x\in(0,1]}\hker_{\play}''(x) > 0$.
We will also say that $\hker_{\play}$ is \emph{steep} when $\lim_{x\to0^{+}}\hker_{\play}'(x) = -\infty$, and, for normalization purposes, we will assume that $\hker_{\play}(0) = 0$.

Suppressing player indices for simplicity, some standard examples of regularizers are as follows:

\begin{example}
[$L^{2}$ regularization]
\label{ex:Eucl}
If $\hker(x) = x^{2}/2$, the players' penalty function is the squared Frobenius norm $\hreg(\denmat) = (1/2)\frobnorm{\denmat}^{2}$, in which case \cref{eq:mirror} gives the orthogonal projector $\mirror(\dmat) = \Eucl_{\denmats}(\dmat) \equiv \argmin_{\denmat\in\denmats} \frobnorm{\denmat - \dmat}$.
This choice leads to the (Frobenius) projection dynamics:
\begin{equation}
\label{eq:PD}
\tag{PD}
\dot\dmat
	= \payfield(\denmat)
	\qquad
\denmat
	= \Eucl_{\denmats}(\dmat)
\end{equation}
\end{example}

\begin{example}
[Von Neumann regularization]
\label{ex:logit}
Another standard choice is $\hker(x) = x\log x$ which yields the (negative) von Neumann entropy $\hreg(\denmat) = \trof{\denmat\log\denmat}$.
By a standard calculation, this choice of regularizer gives rise to the \acdef{MMW} dynamics
\begin{equation}
\label{eq:MMW}
\tag{MMW}
\dot\dmat
	= \payfield(\denmat)
	\qquad
\denmat
	= \frac{\exp(\dmat)}{\trof{\exp(\dmat)}}
\end{equation}
A discrete-time version of these dynamics was introduced in the context of kernel learning by \citet{TRW05} and \citet{KSST12};
for a series of more recent developments in the context of quantum learning see \cite{JW09,ACHK+18,JPS22}.
\end{example}

\begin{example}
[Tsallis regularization]
\label{ex:Tsallis}
Interpolating between the above, the \emph{Tsallis regularizer} is given by $\hker(x) = [\qexp(1-\qexp)]^{-1} (x - x^{\qexp})$ for some $\qexp > 0$ (with the continuity convention $(x-x^{\qexp})/(1-\qexp) = x \log x$ for $\qexp=1$).
When $\qexp \gets 2$, we recover the projection dynamics \eqref{eq:PD};
by contrast, the choice $\qexp \gets 1$ gives rise to \eqref{eq:MMW};
finally,
the choice $\qexp \gets 1/2$ is particularly popular in the context of bandit online learning, \cf \cite{KSST12,ZS21} and references therein.
\end{example}

\para{The mixed-state dynamics of \ac{FTQL}}

Under \eqref{eq:FTQL}, the evolution of the players' mixed states $\stateof{\time}$ is described \emph{implicitly} via that of the auxiliary score matrix $\dstateof{\time}$.
On the other hand, obtaining an \emph{explicit} expression for the dynamics of $\stateof{\time} = \mirror(\dstateof{\time})$ is considerably more difficult because the rules of matrix calculus do not provide an analytic expression for the tensor derivative $\nabla_{\dmat}\mirror(\dmat)$ of $\mirror$, even when the latter is available in closed form.

To circumvent this difficulty, we will work with a unitary eigendecomposition of $\denmat$ of the form
\begin{equation}
\label{eq:eigen}
\denmat
	= \insum_{\pure=1}^{\vdim} \denval_{\pure} \denvec_{\pure} \denvec_{\pure}^{\dag}
\end{equation}
where
$\denval_{\pure} \geq 0$, $\pure=1,\dotsc,\vdim$, is an enumeration of the eigenvalues of $\denmat$,
$\denvec_{\pure} \in \hilbert$ is a unit-norm eigenvector of $\denmat$ corresponding to $\denval_{\pure}$,
and player indices have again been suppressed (to lighten notation).
Since, in general, $\denmat$ does not commute with $\payfield$ (and hence with $\dot\denmat$), the eigenvalues and eigenvectors of $\denmat$ will evolve in a coupled, concurrent manner;
our first result below provides an explicit expression for this co-evolution:

\begin{restatable}{theorem}{dynamics}
\label{thm:dynamics}
Let $\stateof{\time} = \sum_{\pure=1}^{\vdim} \denval_{\pure}(\time) \, \denvec_{\pure}(\time) \denvec_{\pure}^{\dag}(\time)$ be an eigendecomposition of $\stateof{\time}$ as per \eqref{eq:eigen}, and suppose that $\hreg$ is steep.
Then, under \eqref{eq:FTQL}, the entries $[\dot\denmat]_{\pure\purealt} = \denvec_{\pure}^{\dag} \dot\denmat \denvec_{\purealt}$ of $\dot\denmat$ follow the quantum state dynamics:
\begin{equation}
\label{eq:QD}
\tag{QD}
[\dot\denmat]_{\pure\purealt}
	= \begin{dcases}
		\frac{\payent_{\pure\pure}}{\hker''(\denval_{\pure})}
		- \frac
			{\sum_{\puredummy} \payent_{\puredummy\puredummy}/\hker''(\denval_{\puredummy})}
			{\sum_{\puredummy} \hker''(\denval_{\pure})/\hker''(\denval_{\puredummy})}
		&\pure=\purealt
		\\[\medskipamount]
		\frac
			{\denval_{\purealt} - \denval_{\pure}}
			{\hker'(\denval_{\purealt}) - \hker'(\denval_{\pure})}
			\payent_{\pure\purealt}
		&\pure\neq\purealt
	\end{dcases}
\end{equation}
where
$\payent_{\pure\purealt} = \denvec_{\pure}^{\dag} \payfield(\denmat) \denvec_{\purealt}$
and
we are using the continuity convention $(x-y) / \bracks{\hker'(x) - \hker'(y)} = 1/\hker''(x)$ when $y \to x$.
\end{restatable}

The proof of \cref{thm:dynamics} is based on the fact that, if a Hermitian matrix $\mat$ with eigendecomposition $\mat = \sum_{\pure} \matval_{\pure} \matvec_{\pure}\matvec_{\pure}^{\dag}$ follows the dynamics $\dot\mat = \matalt$, a differentiation of the eigenvalue equation $\mat\matvec_{\pure} = \matval_{\pure} \matvec_{\pure}$ yields
\begin{equation}
\label{eq:eig-diff}
\dot\matval_{\pure} \delta_{\pure\purealt}
	= (\matval_{\pure} - \matval_{\purealt})
		\matvec_{\pure}^{\dag} \dot\matvec_{\purealt}
	+ \matvec_{\pure}^{\dag} \matalt \matvec_{\purealt}
	\quad
	\text{for all $\pure,\purealt$}.
\end{equation}
This identity allows us to analyze and derive an expression for $\dot\denmat$ by solving the Lagrangian associated with the maximization problem \eqref{eq:mirror}.
The full proof is relegated to the appendix;
instead, we focus here on a representative example.


\begin{example}
[The \acl{QRD}]
\label{ex:QRD}
An important special case of the dynamics \eqref{eq:QD} is obtained by the von Neumann regularizer $\hreg(\denmat) = \trof{\denmat\log\denmat}$ of \cref{ex:logit}.
This yields the \acli{QRD}
\begin{equation}
\label{eq:QRD}
\tag{QRD}
[\dot\denmat]_{\pure\purealt}
	= \begin{dcases*}
		\txs
		\denval_{\pure} \bracks{\payent_{\pure\pure} - \sum_{\puredummy} \denval_{\puredummy} \payent_{\puredummy\puredummy}}
			&for $\pure=\purealt$
		\\[\smallskipamount]
		\frac{\denval_{\purealt} - \denval_{\pure}}{\log\denval_{\purealt} - \log\denval_{\pure}}
		\payent_{\pure\purealt}
			&for $\pure\neq\purealt$
	\end{dcases*}
\end{equation}
for the diagonal and off-diagonal elements of $\denmat$ respectively.
The diagonal part of \eqref{eq:QRD} is formally analogous to the replicator dynamics of evolutionary game theory \cite{TJ78,Wei95,San10} and captures the evolution of the eigenvalues of $\stateof{\time}$.
Thus, taken together with its off-diagonal component, \eqref{eq:QRD} provides an explicit expression for the evolution of mixed states under \eqref{eq:MMW}.
Alternatively, by applying Fréchet's differentiation formula to \eqref{eq:MMW} directly, \eqref{eq:QRD} can be rewritten in basis-free notation as
\begin{equation}
\label{eq:QRD-free}
\dot\denmat
	= \int_{0}^{1} \denmat^{1-s} \payfield(\denmat) \denmat^{s} \dd s
		- \trof{\denmat \payfield(\denmat)} \denmat
\end{equation}
The dynamics \eqref{eq:QRD} and the coordinate-free expression \eqref{eq:QRD-free} agree with the dynamics of \citet{JPS22} (who derived an equivalent expression under the assumption that $\denmat$ and $\dot\denmat$ commute), but not with the dynamics of \citet{Hid06} that follow a different, unrelated quantization paradigm. 
We provide the relevant calculations in \cref{app:dynamics}.
\endenv
\end{example}

In the sections that follow, we will examine in detail how the classical and quantum components of \eqref{eq:FTQL} interface with each other to determine the player's long-run behavior.

\section{Regret minimization}
\label{sec:regret}

We begin our analysis of \eqref{eq:FTQL} with a result concerning the dynamics' regret minimization properties.
To provide the necessary context, it is known that the \ac{FTRL} dynamics incur at most constant regret in classical, \emph{finite} games \cite{KM17,MPP18}.
As we illustrate below, despite taking place over a \emph{continuum} of pure states, the matrix-valued dynamics \eqref{eq:FTQL} enjoy the same regret minimization guarantees in quantum games.
Formally, we have the following result.

\begin{restatable}{proposition}{regret}
\label{prop:regret}
Let $\stateof{\time} = \mirror(\dstateof{\time})$, $\time\geq\tstart$, be a trajectory of play induced by \eqref{eq:FTQL}.
Then, for all $\horizon\geq\tstart$, we have
\begin{equation}
\label{eq:reg-FTQL}
\reg_{\play}(\horizon)
	\leq \abs{\vdim_{\play} \cdot \hker_{\play}(1/\vdim_{\play}) - \hker_{\play}(1)}.
\end{equation}
\end{restatable}

The proof of \cref{prop:regret} builds on the general theory of \cite{KM17} and is presented in detail in \cref{app:regret}.
Instead, we only note here that the dependence of the bound \eqref{eq:reg-FTQL} on the dimensionality of the game depends crucially on the choice of regularizer:
\begin{enumerate*}
[(\itshape i\hspace*{1pt}\upshape)]
\item
the Euclidean regularizer $\hker(x) = x^{2}/2$ gives an $\bigoh(1)$ bound;
\item
the von Neumann entropy leads to an $\bigoh(\log\vdim)$ dependence;
and, finally,
\item
the bound for the Tsallis regularizer $\hker(x) \propto x - x^{\qexp}$ is $\bigoh(\vdim^{1-\qexp})$ for all $\qexp \in (0,1)$.
\end{enumerate*}
This should be contrasted to \emph{discrete-time} models of online learning, where quadratic regularization leads to suboptimal results relative to \emph{both} the entropic and Tsallis variants \cite{Sli19,LS20,ZS21}.
The reason for this discrete-to-continuous gap has to do with the fact that \eqref{eq:FTQL} admits an \emph{exact} energy function, the so-called \emph{Fenchel coupling}
\begin{align}
\label{eq:Fench}
\fench_{\play}(\denmat_{\play},\dmat_{\play})
	&= \hreg_{\play}(\denmat_{\play})
		+ \hconj_{\play}(\dmat_{\play})
		- \trof{\denmat_{\play}\dmat_{\play}}
\intertext{where $\denmat_{\play}\in\denmats_{\play}$, $\dmat_{\play}\in\dmats_{\play}$, and}
\label{eq:conj}
\hconj_{\play}(\dmat_{\play})
	&= \max\nolimits_{\denmat_{\play}\in\denmats_{\play}}\setof{\trof{\denmat_{\play}\dmat_{\play}} - \hreg_{\play}(\denmat_{\play})}
\end{align}
denotes the convex conjugate of $\hreg_{\play}$.
A version of this primal-dual coupling was first introduced by \cite{MS16} in the setting of finite games,
and it has the following fundamental property:

\begin{restatable}{lemma}{Fenchel}
\label{lem:Fench}
Let $\stateof{\time} = \mirror(\dstateof{\time})$, $\time\geq\tstart$, be a trajectory of play induced by \eqref{eq:FTQL}.
Then, for all $\basemat\in\denmats$, we have
\begin{equation}
\label{eq:dFench}
\ddt\fench_{\play}(\basemat_{\play},\dstateof[\play]{\time})
	= \trof{\payfield_{\play}(\stateof{\time}) \, (\stateof[\play]{\time} - \basemat_{\play})}.
\end{equation}
\end{restatable}

The importance of this lemma (which is proved in \cref{app:Fenchel}) lies in that the \acs{RHS} of \eqref{eq:dFench} is precisely the integrand of the regret, so \cref{prop:regret} is obtained by applying \cref{lem:Fench} to the state that witnesses the maximum in the definition \eqref{eq:regret} of the regret of player $\play$.
We defer the relevant calculations to \cref{app:regret}.

\section{Convergence, stability, and the folk theorem}
\label{sec:folk}

In view of the strong regret minimization guarantees of \cref{prop:regret}, the matrix-valued dynamics \eqref{eq:FTQL} emerge as a very compelling choice from a learning standpoint.
At the same time, even in the case of classical finite games, it is known that regret minimization \emph{does not} suffice to exclude non-rationalizable outcomes:
for example, as was shown by \citet{VZ13}, the players' empirical frequency of play under a no-regret policy may still end up assigning positive selection probability to (strictly) dominated strategies, and \emph{only} dominated ones.
On that account, our aim in the rest of this section will be to take a closer look at the convergence and stability properties of \eqref{eq:FTQL} relative to the game's \aclp{NE}.

\para{Notions of stability and convergence}

Our analysis will require some basic concepts from the theory of dynamical systems, which we quickly discuss below.
To begin with, recall that a \emph{flow} on an abstract metric space $\points$ is a continuous map $\flow \from \R\times\points \to \points$ such that
\begin{enumerate*}
[\upshape(\itshape a\upshape)]
\item
$\flowof{\tstart}{\point} = \point$;
and
\item
$\flowof{\time+\timealt}{\point} = \flowof{\time}{\flowof{\timealt}{\point}}$
\end{enumerate*}
for all $\time,\timealt\in\R$ and all $\point\in\points$.
Informally, a flow is usually generated by the \emph{solution orbits} of a system of well-posed \acp{ODE}, such as \eqref{eq:FTQL}:
in this interpretation, $\flowof{\time}{\point}$ simply denotes the position at time $\time$ of the \ac{ODE} solution that starts at $\point$ at time $\time=\tstart$.

With this in mind, the following notions of invariance and stability will play a key role in our analysis.
Given a point $\base\in\points$, we will say that:
\begin{enumerate}
\item
$\base$ is \emph{stationary} if $\flowof{\time}{\base} = \base$ for all $\time\in\R$.
\item
$\base$ is \emph{Lyapunov stable} \textendash\ or just \emph{stable} \textendash\ if, for every neighborhood $\nhd$ of $\base$ in $\points$, there exists some (smaller) neighborhood $\nhdalt$ of $\base$ in $\points$ such that $\flowof{\time}{\nhdalt} \subseteq \nhd$ for all $\time\in\R$.
In other words, $\base$ is stable if any orbit that starts close enough to $\base$ remains close enough.
\item
$\base$ is \emph{attracting} if it admits a neighborhood $\nhd$ such that $\lim_{\time\to\infty} \flowof{\time}{\point} = \base$ for all $\point\in\nhd$.
In other words, $\base$ is attracting if all nearby orbits converge to $\base$.
\item
$\base$ is \emph{asymptotically stable} if it is stable and attracting.
\end{enumerate}
In what follows, we will seek to characterize precisely the stable and/or attracting states of \eqref{eq:FTQL}.

\para{The classical regime}

To set the stage for the analysis to come, it will be useful to revisit the classical regime of learning in classical \emph{finite} games.
Focusing for concreteness on the standard case of the replicator dynamics (\cf \cref{ex:logit,ex:QRD} above), the stability and convergence landscape for general finite games can be encoded in the so-called ``folk theorem'' of evolutionary game theory, which states the following \cite{Cre03,HS03}:
\begin{enumerate}
\item
\aclp{NE} are stationary.
\item
Limits of interior orbits are \aclp{NE}.
\item
Lyapunov stable states are \aclp{NE}.
\item
A state is asymptotically stable if and only if it is a strict \acl{NE} (that is, every player has a unique best response at equilibrium).
\end{enumerate}
Modulo some technicalities, these properties extend to the entire class of \ac{FTRL} dynamics for learning in classical \emph{finite} games, \cf \cite{MS16,FVGL+20}, and references therein.
However, the nonlinear geometry of the players' state space in quantum games places severe structural limitations on which of these properties transfer over to the non-commutative, matrix-valued setting of \eqref{eq:FTQL}.
We explore this issue below.

\para{Regularized learning in the spectraplex}

A quick look at the \acl{QRD} \eqref{eq:QRD} reveals the following structural property:
an eigenvalue of $\denmat$ that is initially zero in \eqref{eq:QRD} will always remain zero;
likewise, an eigenvalue that is initially positive, will always remain positive.
Formally, this means that the kernel $\ker(\denmat)$ of $\denmat$ remains invariant under \eqref{eq:QRD};
hence, given that the linear span of a Hermitian matrix is the orthocomplement of its kernel, the same holds for $\im(\denmat)$ as well.

The fact that the kernel \textendash\ or, equivalently, the image \textendash\ of a density matrix remains invariant under \eqref{eq:QRD} is the quantum analogue of the fact that the support of a mixed strategy profile remains invariant under the standard replicator dynamics.
In the context of finite games, an immediate consequence of this invariance is that \emph{all} pure strategy profiles are stationary (as zero-dimensional faces of the simplex).
This property extends to \eqref{eq:QRD} and, in fact, to the entire class of mixed-state dynamics under study:
formally,
under \eqref{eq:QD},
\emph{all pure quantum states are stationary.}

That being said, the major qualitative difference between the quantum and classical regimes is that, in quantum games, there is a \emph{continuum} of pure states, namely the entire manifold $\quants$ of rank $1$ density matrices (a product of spheres).
By contrast, in finite games, the pure states are the corners of the simplex $\simplex$ spanned by the player's pure strategies, so they are finite in number and \emph{isolated}.
As a result, in classical finite games, a pure strategy profile \emph{can} be asymptotically stable;
in quantum games, since every pure state is surrounded by other invariant states, \emph{it cannot}.

A second major difference is that, in finite games, strict \aclp{NE} are \emph{robust:}
a small perturbation of the payoffs of the game does not change the game's strict equilibria.
In quantum games, this robustness disappears:
indeed, the variational characterization \eqref{eq:VI} of \aclp{NE} means that $\payfield(\eq)$ must be an element of the normal cone to $\denmats$ at $\eq$;
however, the normal cone to the spectraplex at a matrix of rank $1$ has empty topological interior, so the required membership property cannot be robust (for a graphical illustration, see \cref{fig:pure}).
In particular, any generic perturbation to the payoffs of a quantum game, no matter how small, may lead to a small displacement of the equilibrium in question on the manifold of pure states $\quants$.


\begin{figure}[t]
\centering
\includegraphics[scale=.75]{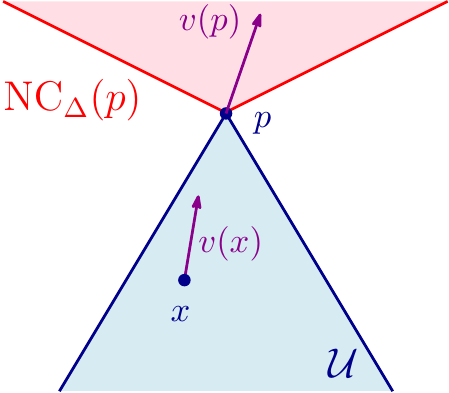}
\hspace{5em}
\includegraphics[scale=.75]{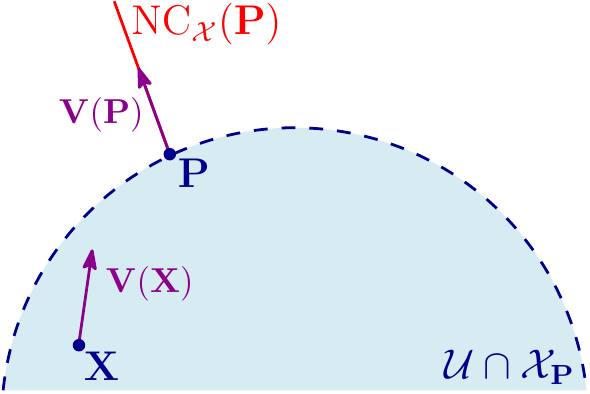}
\caption{The fundamental geometric discrepancy between the classical and quantum regimes (left and right respectively).
In classical finite games, the normal cone $\ncone_{\simplex}(\base)$ to the simplex at a pure strategy $\base$ has nonempty topological interior, so pure \aclp{NE} are generically \emph{robust}:
if $\payv(\base)$ is normal to $\simplex$ at $\base$, it will remain normal to $\simplex$ after a small perturbation.
On the other hand, in quantum games, the normal cone $\ncone_{\denmats}(\basemat)$ to the spectraplex at a pure state $\basemat$ is a ray, so pure \aclp{NE} cannot be robust in this sense.
We also note the different geometry of pure states:
in the simplex, pure strategy profiles are isolated extreme points;
in the spectraplex, pure states form a continuous manifold.}
\label{fig:pure}
\end{figure}


\para{Consistency and variational stability}

In view of the above, we can already draw two major conclusions for the quantum setting:
\begin{enumerate}
\item
Any concept of asymptotic stability must also include a notion of \emph{consistency:}
a state cannot be accessed if it is absent from the linear span of the dynamics' initial state.
\item
Any concept of robustness must likewise incorporate a notion of \emph{variational stability:}
small perturbations to an equilibrium state must tend to reinstate it.
\end{enumerate}
We formalize these two ideas as follows:
\smallskip

\begin{definition}
\label{def:consistent}
Fix a state $\basemat\in\denmats$ and let
\begin{equation}
\label{eq:support}
\denmats_{\basemat}
	\defeq \setdef{\denmat\in\denmats}{\ker(\denmat) \leq \ker(\basemat)}
\end{equation}
denote the \emph{domain of consistency} of $\basemat$ in $\denmats$, \ie the set of mixed states whose linear span contains that of $\basemat$.
Then, given a flow $\pflow \from \R\times\denmats \to \denmats$, we will say that:
\begin{enumerate}
\item
$\basemat$ is \emph{consistently attracting} if it admits a neighborhood $\pnhd$ such that $\lim_{\time\to\infty} \pflowof{\time}{\denmat} = \basemat$ for all $\denmat\in\pnhd\cap\denmats_{\basemat}$.
In other words, $\basemat$ is consistently attracting if it attracts all nearby consistent initializations.
\item
$\basemat$ is \emph{consistently asymptotically stable} if it is Lyapunov stable and consistently attracting.
\end{enumerate}
\end{definition}

\begin{definition}
\label{def:VS}
We say that $\eq\in\denmats$ is \emph{\textpar{locally} variationally stable} if there exists a neighborhood $\pnhd$ of $\eq$ in $\denmats$ such that
\begin{equation}
\label{eq:VS}
\tag{VS}
\trof{\payfield(\denmat) (\denmat - \eq)}
	< 0
	\quad
	\text{for all $\denmat\in\pnhd\exclude{\eq}$}.
\end{equation}
\end{definition}

Intuitively, \cref{def:consistent} captures precisely the accessibility condition that we discussed above:
$\denmats_{\basemat}$ is a dense convex set consisting of the relative interior of all faces of $\denmats$ that contain $\basemat$ (including $\denmats$ itself).
As for \cref{def:VS}, variational stability should be seen as an equilibrium refinement in the spirit of the seminal concept of \emph{evolutionary stability} \cite{MP73,May82} that underlies the ``folk theorem'' for finite games;
for a detailed discussion of \eqref{eq:VS} in the context of continuous games, see \cite{MZ19}.

\begin{remark*}
In \cref{def:consistent}, the notion of Lyapunov stability does not have a ``consistency'' caveat tacked on.
As we discuss in \cref{app:folk}, the reason for this is that, in the case of Lyapunov stability, the two notions end up coinciding, so it is not necessary to make this distinction.
\endenv
\end{remark*}

\para{A quantum ``folk theorem''}

We are now in a position to state the main result of this section.
To simplify the presentation, we assume below that \eqref{eq:FTQL} is run with a steep regularizer and, as per \cref{thm:dynamics}, the quantum state dynamics \eqref{eq:QD} refer to the flow induced by \eqref{eq:FTQL} on $\denmats$.

\begin{restatable}{theorem}{folk}
\label{thm:folk}
Let $\qgame \equiv \qgamefull$ be a quantum game, fix some state $\eq\in\denmats$, and let $\stateof{\time} = \mirror(\dstateof{\time})$ be a trajectory of play induced by \eqref{eq:FTQL}
with steep regularizers.
Then:
\begin{enumerate}
\item
If $\eq$ is a \acl{NE}, it is a rest point of \eqref{eq:QD}.
\item
If $\stateof{\time} \to \eq$ as $\time\to\infty$, $\eq$ is a \acl{NE} of $\qgame$.
\item
If $\eq$ is stable under \eqref{eq:QD}, it is a \acl{NE} of $\qgame$.
\item
If $\eq$ is consistently asymptotically stable, then it is pure.
\item
If $\eq$ satisfies \eqref{eq:VS}, it is consistently asymptotically stable.
\end{enumerate}
\end{restatable}

Before discussing the proof of \cref{thm:folk}, some remarks are in order.
Perhaps the most important one concerns the asymptotic stability part of the theorem (which is arguably the most salient point of the classical folk theorem as well).
Here, even though the standard notion of asymptotic stability is ruled out by the geometry of the game's state space, \eqref{eq:FTQL} achieves the next best thing:
by definition, states that are consistently asymptotically stable attract all but a measure zero of nearby initial conditions, and \cref{thm:folk} shows that \emph{only} pure states can have this property.
This selection result has important implications for quantum games because it shows that regularized learning essentially ``collapses'' an initial mixed state to a \emph{specific} pure state \textendash\ and this, despite the fact that any mixed state can be prepared by an infinitum of combinations of pure states.

On the flip side of all this, the implication that variationally stable states are also (consistently) asymptotically stable provides a relevant convergence criterion for \eqref{eq:FTQL} and indicates an inherent robustness to variations of player beliefs and predictions.
In particular, since \eqref{eq:VS} only involves the primitives of the underlying game, the fact that such states are attracting under \emph{all} \ac{FTQL} dynamics means that they can be seen as universal attractors \textendash\ and since only pure states can have this property, we also infer indirectly that variationally stable states are \emph{a fortiori} pure.


\begin{figure}[t]
\centering
\includegraphics[height=.25\textheight]{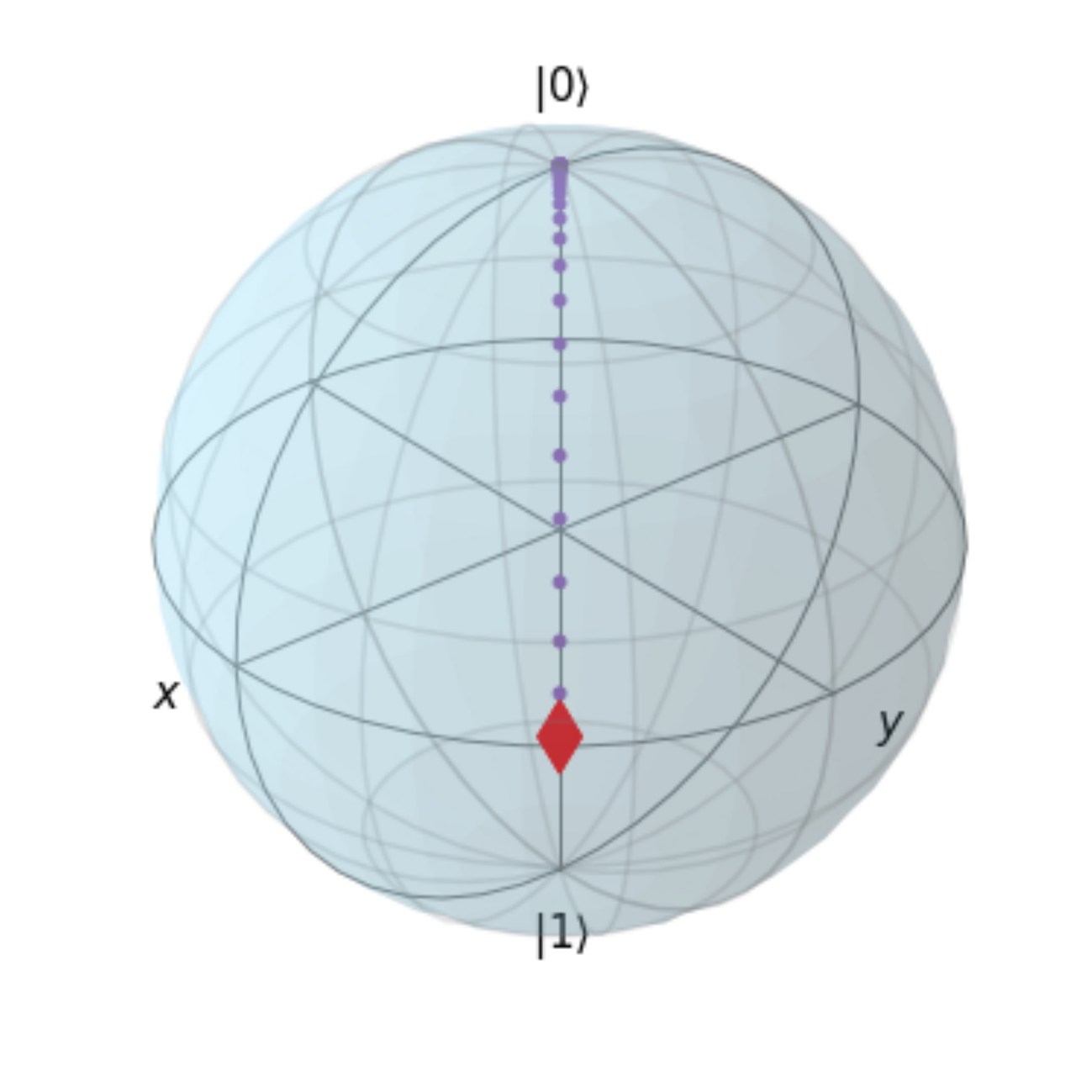}
\hspace{5em}
\includegraphics[height=.25\textheight]{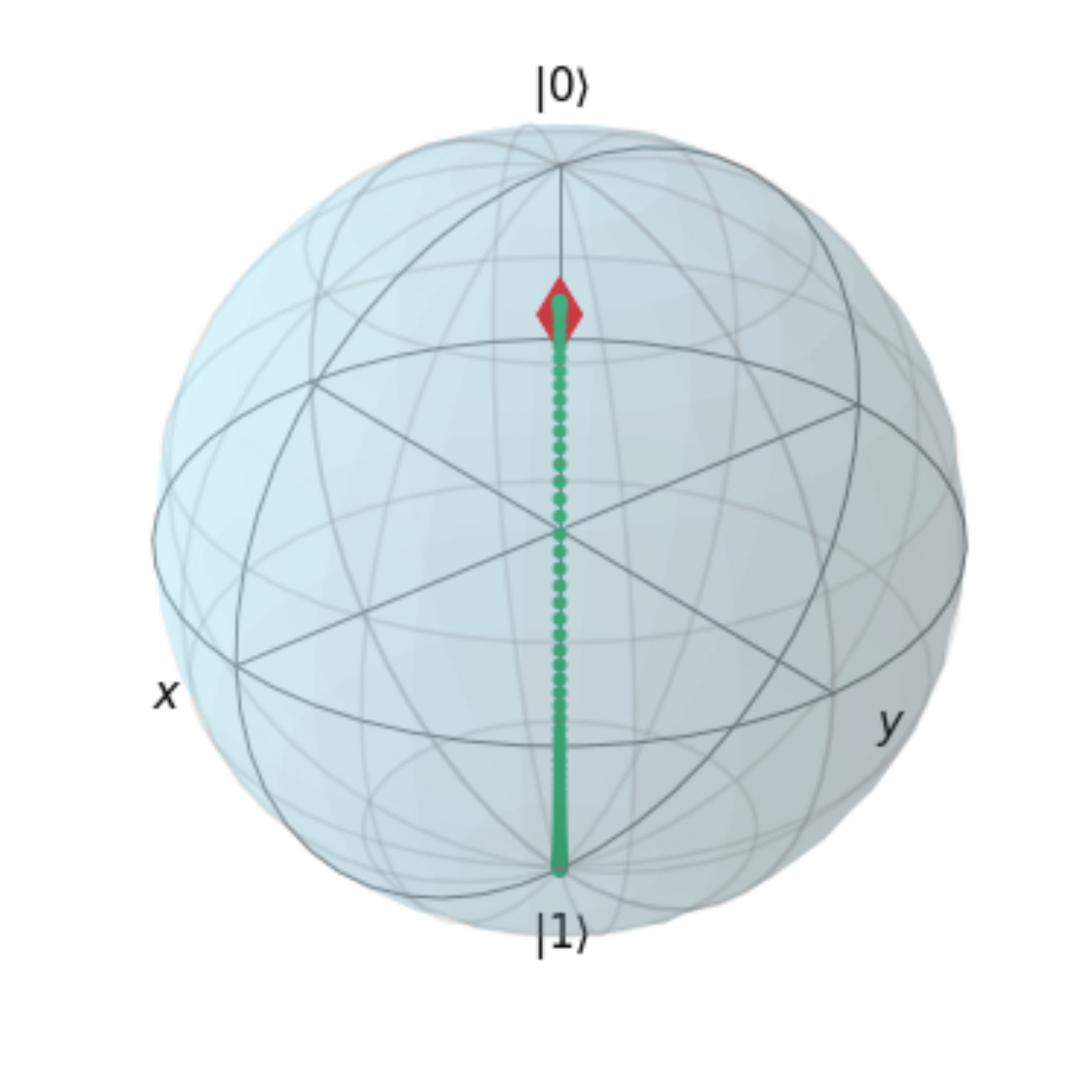}
\caption{Convergence to a variationally stable equilibrium in a quantum anti-coordination game.
Each player's trajectory is represented in their individual Bloch sphere;
see \cref{app:numerics} for the details.}
\vspace{-\baselineskip}
\label{fig:traj}
\end{figure}


\para{Proof sketch}

%

We conclude this section with a short outline of the proof of the last two parts of \cref{thm:folk} (the full proof is presented in \cref{app:folk}).
First, the implication that only pure states can be stable attractors hinges on the following volume preservation principle:

\begin{proposition}
\label{prop:volume}
Let $\dnhd \subseteq \dmats$ be an open set of initial conditions of \eqref{eq:FTQL}, and let $\dnhd_{\time} = \dflowof{\time}{\dnhd}$, $\time\geq0$, denote the evolution of $\dnhd$ under the flow $\dflow \from \R\times\dmats \to \dmats$ of \eqref{eq:FTQL}.
Then, $\vol(\dnhd_{\time}) = \vol(\dnhd)$ for all $\time\geq0$.
\end{proposition}

\cref{prop:volume} is central to our work in several aspects.
First, this volume-preservation property applies to \emph{all} quantum games and \emph{all} \ac{FTQL} dynamics, generalizing in this way the very recent analysis of \citet{JPS22} for \eqref{eq:MMW}\,/\,\eqref{eq:QRD} in two-player, zero-sum games.
This degree of generality is due to the fact that each player's payoff function is individually linear in the player's own density matrix, so the individual gradient fields $\payfield_{\play}$ do not depend on $\denmat_{\play}$;
\cref{prop:volume} then follows by an application of Liouville's formula.

In this regard, \cref{prop:volume} can be seen as a quantum analogue of the universal volume-preservation property of \ac{FTRL} dynamics in \emph{finite} games \cite{CGM15,FVGL+20}.
However, as in the classical world, it is worth noting that \cref{prop:volume} concerns the flow of \eqref{eq:FTQL} on $\dmats$, not the induced flow on $\denmats$;
in particular, $\mirror(\dnhd_{\time})$ could still collapse to a boundary state of $\denmats$ if $\dnhd_{\time}$ escapes to infinity.
Because of this, although \cref{prop:volume} precludes convergence to full-rank equilibria relatively easily, excluding lower-rank equilibria requires more delicate arguments, where the notion of consistency plays a major role (and has no classical counterpart).

Finally, the last part of \cref{thm:folk} leverages an energy argument in the spirit of Lyapunov's direct method.
In tune with the primal-dual nature of \eqref{eq:FTQL}, a natural choice of energy function is the Fenchel coupling $\energy(\time) = \fench(\eq,\dstateof{\time})$ of \eqref{eq:Fench}, which is in turn linked to \eqref{eq:VS} via \cref{lem:Fench}.
Putting these two elements together readily yields that $\lim_{\time\to\infty} \energy(\time)$ exists;
then, by a trapping argument, it can be shown that
there exists a sequence of times $\time_{\run}\to\infty$ such that $\stateof{\time_{\run}} \to \eq$,
which allows us to conclude that $\energy(\time) \to 0$ and ultimately yields our claim.

\section{Poincaré recurrence in zero-sum games}
\label{sec:recurrence}

In this last section, our aim is to study the long-run behavior of \eqref{eq:FTQL} in zero-sum games, a class of games that arise frequently in applications \textendash\ and is possibly the most widely studied family of quantum games.

In classical finite games, it is well known that the replicator dynamics \textendash\ and, more generally, all \ac{FTRL} dynamics \textendash\ produce trajectories that cycle around interior equilibria, leading to a behavior known as \emph{Poincaré recurrence} \cite{PS14,MPP18}.
Formally, in our setting, the dynamics \eqref{eq:FTQL} are said to be Poincaré recurrent if, for almost every initial condition $\stateof{\tstart} = \mirror(\dstateof{\tstart}) \in \denmats$ (in the sense of Lebesgue), there exists an increasing sequence of times $\curr[\time] \to \infty$ such that $\stateof{\curr[\time]} \to \stateof{\tstart}$ as $\run\to\infty$, \ie almost all trajectories of \eqref{eq:FTQL} return infinitely close to their starting point infinitely often.

In our notation, a two-player zero-sum quantum game $\qgame \equiv \qgamefull$ corresponds to the case where $\players = \{1,2\}$ and $\pay_{1} = -\pay_{2} = -\minmax$ for some min-max merit function $\minmax \from \denmats_{1}\times\denmats_{2} \to \R$.
In this context, \citet{JPS22} recently showed that the \acl{MMW} dynamics \eqref{eq:MMW} are Poincaré recurrent in all two-player, zero-sum games that admit a full-rank equilibrium.
As we show below, this property extends to the \emph{entire} class of regularized learning dynamics under study.

\begin{restatable}{theorem}{recurrence}
\label{thm:recurrence}
Let $\qgame \equiv \qgamefull$ be a $2$-player zero-sum game admitting a full-rank equilibrium $\eq\in\relint\denmats$.
Then, almost every solution orbit $\stateof{\time} = \mirror(\dstateof{\time})$ of \eqref{eq:FTQL} is Poincaré recurrent:
there exists an increasing sequence of times $\curr[\time] \to \infty$ such that $\lim_{\run\to\infty} \stateof{\curr[\time]} = \stateof{\tstart}$. 
\end{restatable}

The proof of \cref{thm:recurrence} comprises three main ingredients:
The first is \cref{prop:volume}, which shows that the flow induced by \eqref{eq:FTQL} on $\dmats$ is volume-preserving.
The second is that the Fenchel coupling \eqref{eq:Fench} that played the role of a local energy function in \cref{thm:folk} becomes a \emph{constant of motion} in zero-sum games (this is also where the full-rank assumption comes into play).
The third is that, modulo a certain quotient process that does not affect the dynamics' trajectories on $\denmats$, the level sets of the Fenchel coupling are bounded;
this is the most challenging part of the proof, and the point where the non-commutativity of the various processes involved complicates things considerably.
With these basic ingredients in place, our result follows by an application of Poincaré's theorem, which states that volume-preserving flows with bounded trajectories are recurrent.
We provide the relevant details in \cref{app:recurrence}.

\section{Concluding remarks}
\label{sec:discussion}

When quantum computing models are deployed in a multi-agent context \textendash\ from autonomous vehicles to quantum \acsp{GAN} \textendash\ the players' interaction landscape changes dramatically relative to classical interactions.
The study of game-theoretic learning in this quantum setting is still in its infancy, so it is not clear at this stage what can be expected by quantum players with bounded rationality.
In this regard, the study of the \ac{FTQL} dynamics provides the following important insights:
the geometric structure of quantum state space leads to an inflation of ``learning traps'' (stationary states) that have no classical counterpart;
nonetheless, the only states that can be stable and attracting under \ac{FTQL} are the game's pure quantum equilibria.
Solidifying our understanding of the limits of quantum game-theoretic learning is a particularly fruitful research direction with potentially far-reaching implications for the deployment of multi-agent quantum computing systems.




\appendix
\numberwithin{equation}{section}		
\numberwithin{lemma}{section}		
\numberwithin{proposition}{section}		
\numberwithin{theorem}{section}		
\numberwithin{corollary}{section}		

\section{Regularized best responses and the Fenchel coupling}
\label{app:Fenchel}

\renewcommand{\vecspace}{\boldsymbol{\mathcal{V}}}
\newmacro{\dval}{\dmat}
\newmacro{\dvals}{\mathbf{\Psi}}
\newmacro{\dvecs}{\mathbf{U}}

\newacro{lsc}[l.s.c.]{lower-semicontinuous}

In this appendix, we introduce the basic properties of the regularized best response map $\mirror$ and the Fenchel coupling.
To simplify notation, we will suppress player indices throughout;
formally, this means that, for example,
$\denmat$ stands for the profile $\denmat = (\denmat_{\play})_{\play\in\players}$;
the ensemble mirror map $\mirror\from\dmats\to\denmats$ denotes the product map $\mirror = \prod_{\play} \mirror_{\play}$;
the aggregate regularizer $\hreg$ on $\denmats$ represents the sum $\hreg = \sum_{\play} \hreg_{\play}$, etc.
The proper substitutions are easily inferred from the context, so there is no danger of confusion.

We will also assume that $\pspace$ is endowed with some abstract norm $\norm{\cdot}$, not necessarily induced by the Hilbert structure of the quantum state space $\hilbert$.
Also, for technical reasons, it will also be convenient to view $\hreg$ as an extended-real-valued function $\hreg\from\pspace\to\R \cup \{\infty\}$ defined over the entire space $\pspace \equiv \herm[\vdim]$ of $\vdim\times\vdim$ Hermitian matrices by assigning the value $\hreg(\denmat) = +\infty$ to all $\denmat\in\pspace\setminus\denmats$.
Following \citet{Roc70}, this allows us to define the \emph{subdifferential} of $\hreg$ at $\denmat\in\denmats$ as
\begin{equation}
\label{eq:subdiff}
\subd\hreg(\denmat)
	\defeq \setdef
		{\dmat\in\dmats}
		{\hreg(\denmatalt) \geq \hreg(\denmat) + \trof{\dmat (\denmatalt - \denmat)} \; \text{for all $\denmatalt\in\pspace$}},
\end{equation}
where $\dmats \equiv \dmats$ plays here the role of the algebraic dual of $\pspace$.
The \emph{domain of subdifferentiability} of $\hreg$ is then defined as
\begin{equation}
\label{eq:subdom}
\dom\subd\hreg
	\defeq \setdef{\denmat\in\dom\hreg}{\subd\hreg \neq \varnothing}
\end{equation}
and the \emph{convex conjugate} of $\hreg$ is given by the expression
\begin{equation}
\label{eq:hconj}
\hconj(\dmat)
	= \max_{\denmat\in\denmats} \{ \trof{\dmat \denmat} - \hreg(\denmat) \}
	\quad
	\text{for all $\dmat\in\dmats$}.
\end{equation}

Since $\hreg$ is a trace function of the form $\hreg(\denmat) = \trof{\hker(\denmat)}$ with $\hker$ strongly convex over $[0,1]$, it readily follows that $\hreg$ is also $\hstr$-strongly convex relative to $\norm{\cdot}$ for some positive constant $\hstr > 0$.
We then have the following basic results:

\begin{lemma}
\label{lem:mirror}
Let $\hreg$ be a $\hstr$-strongly convex regularizer on $\denmats$ as above, and let $\mirror\from\dmats\to\denmats$ be its induced mirror map.
Then:
\begin{enumerate}
\item
$\mirror$ is single-valued on $\dmats$:
in particular, for all $\denmat\in\denmats$, $\dmat\in\dmats$, we have $\denmat = \mirror(\dmat) \iff \dmat \in \subd\hreg(\denmat)$.
\item
The prox-domain $\proxdom \defeq \im\mirror$ of $\hreg$ satisfies $\relint\denmats \subseteq \proxdom \subseteq \denmats$.
\item
$\mirror$ is $(1/\hstr)$-Lipschitz continuous and $\mirror = \nabla\hconj$.
\end{enumerate}
\end{lemma}

Variants of these properties can be found in several points in the literature \textendash\ see \eg \cite{MZ19} and references therein \textendash\ but we provide a few quick pointers here for completeness.

\begin{proof}
For the first property of $\mirror$, note that the maximum in \eqref{eq:mirror} is attained for all $\dmat\in\dmats$ because $\hreg$ is \ac{lsc} and strongly convex.
Furthermore, by Fermat's theorem for stationarity points, $\denmat$ solves \eqref{eq:mirror} if and only if $\dmat - \subd\hreg(\denmat) \ni 0$, \ie if and only if $\dmat\in\subd\hreg(\denmat)$.
The above also shows that $\proxdom = \dom\subd\hreg$;
since $\relint\denmats \subseteq \dom\subd\hreg \subseteq \denmats$ \citep[Chap.~26]{Roc70}, our second claim follows.
Finally, for our third claim, the expression $\mirror = \nabla\hconj$ is an immediate consequence of Danskin's theorem \cite{SDR09}, while the Lipschitz continuity of $\mirror$ follows from standard variational analysis results, \citep[see \eg][Theorem 12.60(b)]{RW98}.
\end{proof}

The next properties of $\mirror$ are more specialized and are intimately related to the structure of the spectraplex:

\begin{lemma}
\label{lem:spectraplex}
With notation and assumptions as in \cref{lem:mirror}, we have:
\begin{enumerate}
\item
$\subd\hreg(\denmat) = \hker'(\denmat) + \ncone(\denmat)$, where $\ncone(\denmat)$ denotes the normal cone to $\denmats$ at $\denmat$;
in particular, for all $\denmat\in\relint\denmats$, we have $\subd\hreg(\denmat) = \setdef{\hker'(\denmat) + \coef \eye}{\coef\in\R}$.
\item
$\dmat$ and $\mirror(\dmat)$ commute for all $\dmat\in\dmats$.
\item
If $\eig(\dmat) = \{\dval_{\coord}\}_{\coord = 1}^{\vdim}$ and $\eig(\denmat) = \{\denval_{\coord}\}_{\coord = 1}^{\vdim}$ respectively denote the eigenvalues of $\dmat\in\dmats$ and $\denmat = \mirror(\dmat) \in \denmats$, we have $\denval_{\coord} \to 0$ whenever $\dval_{\coord} - \dval_{\coordalt} \to -\infty$ for some $\coordalt\neq\coord$.
\end{enumerate}
\end{lemma}

\begin{proof}
For our first claim, let $\delta_{\denmats}$ denote the (convex) indicator of $\denmats$, \viz $\delta_{\denmats}(\denmat) = 0$ for all $\denmat\in\denmats$ and $\delta_{\denmats}(\denmat) = \infty$ for all $\denmat\in\pspace\setminus\denmats$.
By standard convex analysis arguments \cite{Roc70}, we have $\subd\delta_{\denmats}(\denmat) = \ncone(\denmat)$ for all $\denmat\in\denmats$, so
\begin{equation}
\subd\hreg(\denmat)
	= \subd \bracks*{\trof{\hker(\denmat)} + \delta_{\denmats}(\denmat)}
	= \hker'(\denmat) + \ncone(\denmat),
\end{equation}
and our assertion follows.

For our second claim, let $\dmat = \dvecs \dvals \dvecs^{\dag}$ be an eigen-decomposition of $\dmat$.
Then the objective of \eqref{eq:mirror} can be rewritten as
\begin{align}
\trof{\dmat \denmat} - \hreg(\denmat)
	= \trof{\dvecs \dvals \dvecs^{\dag} \denmat} - \trof{\hker(\denmat)}
	&= \trof{\dvals \dvecs^{\dag} \denmat \dvecs} - \trof{\dvecs^{\dag} \hker(\denmat) \dvecs}
	\notag\\
	&= \trof{\dvals \dvecs^{\dag} \denmat \dvecs} - \trof{\hker(\dvecs^{\dag} \denmat \dvecs)}
\end{align}
since $\denmat$ and $\dvecs^{\dag} \denmat \dvecs$ have the same eigenvalues and $\trof{\hker(\cdot)}$ only depends on the eigenvalues of its argument.
Since $\denmats$ remains invariant under conjugation by any unitary matrix (\ie $\dvecs^{\dag} \denmats \dvecs = \denmats$ for every unitary $\dvecs$), it follows that
\begin{equation}
\max_{\denmat\in\denmats} \braces{\trof{\dmat \denmat} - \hreg(\denmat)}
	= \max_{\denmat\in\denmats} \braces{\trof{\dvals \denmat - \hker(\denmat)}}
	= \max_{\denval\in\simplex_{\vdim}}
		\braces*{\sum_{\coord=1}^{\vdim} \bracks{\dval_{\coord}\denval_{\coord} - \hker(\denval_{\coord})}}
\end{equation}
where, in obvious notation, $\denval_{\coord}$ and $\dval_{\coord}$ respectively denote the eigenvalues of $\denmat$ and $\dmat$.
This expression shows that the maximum value of \eqref{eq:mirror} is attained when $\dmat$ and $\denmat$ have a common basis of eigenvectors, which in turn means that they commute.
Since $\hker$ is strongly convex by assumption, \eqref{eq:mirror} admits a unique solution, and our claim follows.

Finally, for our third claim, let $\denmat_{\run} = \mirror(\dmat_{\run})$ for some sequence $\dmat_{\run} \in \dmats$, and write $\denval_{\coord,\run}$ and $\dval_{\coord,\run}$ for the eigenvalues of $\denmat_{\run}$ and $\dmat_{\run}$ respectively.
We seek to show that $\denval_{\coord,\run} \to 0$ if there exists some $\coordalt$ such that $\dval_{\coord,\run} - \dval_{\coordalt,\run} \to -\infty$.

Now, by descending to a subsequence if necessary, we may assume without loss of generality that there exists some $\eps>0$ such that $\denval_{\pure,\run} \geq \eps > 0$ for all $\run$.
Then, by the defining relation \eqref{eq:mirror} of $\mirror(\dmat)$, we have
\begin{flalign}
\label{eq:Qcomp1}
\sum_{\coord=1}^{\vdim} \bracks{\dval_{\coord,\run} \denval_{\coord,\run} - \hker(\denval_{\coord,\run})}
	\geq \sum_{\coord=1}^{\vdim} \bracks{\dval_{\coord,\run} \alt\denval_{\coord,\run} - \hker(\alt\denval_{\coord,\run})}
\end{flalign}
for all $\alt\denval\in\simplex_{\vdim}$.
Therefore, taking $\denval_{\run}' = \denval_{\run} + \eps(\bvec_{\purealt} - \bvec_{\pure})$, we readily obtain
\begin{equation}
\label{eq:Qcomp2}
\eps (\dval_{\pure,\run} - \dval_{\purealt,\run})
	\geq \sum_{\kappa} \bracks{\hker(\denval_{\kappa,\run}) - \hker(\alt\denval_{\kappa,\run})}
	\geq \min\hreg - \max\hreg
\end{equation}
which contradicts our original assumption that $\dval_{\pure,\run} - \dval_{\purealt,\run} \to -\infty$.
With $\denmats$ compact, the above shows that $\denval_{\pure}^{\ast} = 0$ for any limit point $x^{\ast}$ of $\denval_{\run}$, as claimed.
\end{proof}

With all this in hand, we proceed below with the basic properties of the Fenchel coupling \textendash\ so named to account for the fact that it collects all terms of Fenchel's inequality, \cf \cite{MZ19}.
The two main properties we will need are encoded in the following proposition

\begin{proposition}
\label{prop:Fenchel}
With notation and assumptions as in \cref{lem:mirror}, we have:
\begin{enumerate}
\item
$\fench(\basemat,\dmat) \geq (\hstr/2) \, \norm{\mirror(\dmat) - \basemat}^{2}$
	for all $\basemat\in\denmats$, $\dmat\in\dmats$.
\item
$\nabla_{\dmat^{\top}} \fench(\basemat,\dmat) = \mirror(\dmat) - \basemat$
	for all $\basemat\in\denmats$, $\dmat\in\dmats$.
\end{enumerate}
\end{proposition}

\begin{proof}
For our first claim, let $\denmat=\mirror(\dmat)$ so $\hconj(\dmat) = \trof{\dmat \denmat} - \hreg(\denmat)$ by the definition \eqref{eq:hconj} of $\hconj$.
The strong convexity of $\hreg$ then yields
\begin{flalign}
\hreg(\denmat) + t \trof{\dmat (\basemat - \denmat)}
	\leq \hreg(\denmat + t(\basemat - \denmat))
	\leq t \hreg(\basemat) + (1-t) \hreg(\denmat) - \tfrac{1}{2} \hstr t(1-t) \norm{\denmat - \basemat}^{2},
\end{flalign}
leading to the bound
\begin{equation}
\label{eq:divbound}
\tfrac{1}{2} \hstr(1-t) \norm{\denmat - \basemat}^{2}
	\leq \hreg(\basemat) - \hreg(\denmat) - \trof{\dmat (\basemat - \denmat)}
	= \fench(\basemat,\dmat)
\end{equation}
for all $t\in(0,1]$.
Our assertion then follows by letting $t\to0^{+}$ in \eqref{eq:divbound}.

As for our second claim, \cref{lem:mirror} readily yields
\begin{flalign}
\nabla\fench(\basemat,\dmat)
	= \nabla_{\dmat^{\top}} \hconj(\dmat) - \nabla_{\dmat^{\top}} \trof{\dmat\basemat}
	= \mirror(\dmat) - \basemat
\end{flalign}
as claimed.
\end{proof}

We are now in a position to prove \cref{lem:Fench} on the evolution of the Fenchel coupling under the dynamics \eqref{eq:FTQL};
for convenience, we restate the result below.

\Fenchel*

\begin{proof}
By definition, we have
\begin{flalign}
\label{eq:dF}
\frac{d}{dt} \fench_{\play}(\basemat_{\play},\dstateof[\play]{\time})
	&= \ddt \bracks*{\hreg_{\play}(\basemat_{\play}) + \hconj_{\play}(\dstateof[\play]{\time}) - \trof{\dstateof[\play]{\time} \basemat_{\play}} }
	\notag\\
	&= \trof{\dot\dstate_{\play}(\time) \, (\stateof[\play]{\time} - \basemat_{\play})}
	= \trof{\payfield_{\play}(\stateof{\time}) \, (\stateof{\time} - \basemat)},
\end{flalign}
where, in the last line, we used \cref{prop:Fenchel}.
\end{proof}

The last result we will need is a ``reciprocity'' property in the spirit of \cite{MZ19} which shows that the convergence topology induced by $\fench$ on $\denmats$ is compatible with the ordinary one.
While this property is relatively easy to establish in linear polytopes with decomposable regularizers, the matrix setting complicates things considerably.

\begin{proposition}
\label{prop:reciprocity}
Let $\dmat_{\run}$, $\run=\running$, be a sequence in $\dmats$.
Then $\denmat_{\run} = \mirror(\dmat_{\run})$ converges to $\basemat \in \denmats$ if and only if $\lim_{\run\to\infty} \fench(\basemat,\dmat_{\run}) = 0$.
\end{proposition}

\begin{proof}
We begin by showing the direct implication ``$\implies$''.
Indeed, assuming that $\denmat_{\run} \to \basemat$, the definition of $\fench$ gives
\begin{align}
\fench(\basemat,\denmat_{\run})
	&= \hreg(\basemat)
		+ \hconj(\dmat_{\run})
		- \trof{\dmat_{\run} \basemat}
	\notag\\
	&= \hreg(\basemat) - \hreg(\denmat_{\run})
		+ \trof{\dmat_{\run} \, (\denmat_{\run} - \basemat)}
	\notag\\
	&= \hreg(\basemat) - \hreg(\denmat_{\run})
		+ \trof{\hker'(\denmat_{\run}) (\denmat_{\run} - \basemat)}	
\end{align}
where, in the last equality, we used the first part of \cref{lem:spectraplex}.
Since $\denmat_{\run} \to \basemat$, the term $\hreg(\basemat) - \hreg(\denmat_{\run})$ vanishes as $\run\to\infty$, so it suffices to show the same for the second term as well.
To that end, if we let $\denmat_{\run} = \sum_{\coord} \denval_{\coord,\run} \denvec_{\coord,\run} \denvec_{\coord,\run}^{\dag}$ be an eigen-decomposition of $\denmat_{\run}$ and we set $\base_{\coord,\run} = \denvec_{\coord,\run}^{\dag} \basemat \denvec_{\coord,\run}$ and $\base_{\coord} = \lim_{\run\to\infty} \base_{\coord,\run}$ (the limit exists because $\denmat_{\run}$ converges to $\basemat$ by assumption), we readily get
\begin{align}
\trof{\hker'(\denmat_{\run}) (\denmat_{\run} - \basemat)}
	&= \sum_{\coord=1}^{\vdim} \hker'(\denval_{\coord,\run}) (\denval_{\coord,\run} - \base_{\coord,\run})
	\notag\\
	&= \sum_{\coord\in\supp(\base)} \hker'(\denval_{\coord,\run}) (\denval_{\coord,\run} - \base_{\coord,\run})
		+ \sum_{\coord:\base_{\coord} = 0} \hker'(\denval_{\coord,\run}) (\denval_{\coord,\run} - \base_{\coord,\run})
\end{align}
Since $\denmat_{\run}$ converges to $\basemat$ the first sum above vanishes:
this is because $\lim_{\run\to\infty} \denval_{\pure,\run} > 0$, so $\hker'(\denval_{\coord,\run})$ has a finite limit for all $\coord\in\supp(\base)$.
As for the second, since $\denval_{\coord,\run} - \base_{\coord,\run} \leq \denval_{\coord,\run}$, it suffices to show that $z \hker'(z) \to 0$ as $z\to0^{+}$;
we prove that this is so in \cref{lem:growth} below, for $f = \hker'$.

Finally, as for the converse implication ``$\Longleftarrow$'', our assertion follows immediately from the first part of \cref{prop:Fenchel}.
\end{proof}

\begin{lemma}
\label{lem:growth}
Let $f:(0,1] \to \R$ be differentiable such that (i) $\lim_{x\to 0^+} f(x) = -\infty$, (ii) $\inf_{x \in (0,1]} f'(x) > 0$, and (iii) $\int_0^1 \abs{f(x)}\dd x < \infty$. Then, $\lim_{x\to 0^+} xf(x) = 0$.
\end{lemma}

\begin{proof}
For convenience, let $g = -f$. Then, we readily get (i) $\lim_{x\to 0^+} g(x) = +\infty$, (ii) $\inf_{x \in (0,1]} g'(x) < 0$, and (iii) $\int_0^1 \abs{g(x)}\dd x < \infty$. Since $\inf_{x \in (0,1]} g'(x) < 0$, we get that $g$ is strictly decreasing. Now, $\lim_{x\to 0^+} g(x) = +\infty$ implies that $g$ is strictly positive close to $0$, \ie $g(x) > 0$ in $(0,c)$ for some $c > 0$. Let $\eps \in (0,c)$. Since $g$ is decreasing and positive in $(0,c)$, we have:
\begin{equation}
\label{eq:sandwich}
0 \leq \eps g(\eps) \leq \int_0^{\eps}g(x)\dd x
\end{equation}
We will now show that $\lim_{\eps\to 0^+}\int_0^{\eps}g(x)\dd x = 0$. For this, we can write $\int_0^{\eps}g(x)\dd x$ as:
\begin{equation}
\label{eq:Lebesgue_notation}
\int_0^{\eps}g(x)\dd x = \int\one_{(0,\eps)}g(x)\dd x
\end{equation}
where $\one_{A}(x) = 1$ if $x \in A$ and $0$, otherwise.
Let $\setof{\eps_n}_{n \in \N}$ be a sequence of positive reals with $\eps_n \to 0$ as $n \to \infty$. Then, it holds:
\begin{equation}
    \one_{(0,\eps_n)}g(x) \xrightarrow{n \to \infty} 0 \quad \text{for all } x \in (0,1]
\end{equation}
since, for fixed $x \in (0,1]$, we have that $\eps_n < x$ for all $n$ large enough, which implies that $\one_{(0,\eps_n)}g(x) = 0$. Moreover, for all $x \in (0,1]$, it holds:
\begin{equation}
    \abs{\one_{(0,\eps_n)}g(x)} \leq \abs{g(x)}
\end{equation}
with $\int_0^1 \abs{g(x)}\dd x < \infty$.
Hence, by the dominated convergence theorem \cite{Fol99}, we get that:
\begin{equation}
    \lim_{n\to\infty}\int_0^1 \one_{(0,\eps_n)}g(x)\dd x = 0
\end{equation}
and, since $\setof{\eps_n}_{n\in\N}$ was arbitrary, we conclude that $\lim_{\eps\to 0^+}\int_0^{\eps}g(x)\dd x = 0$. Hence, combining it with \eqref{eq:sandwich}, we get that $\lim_{x\to 0^+} xg(x) = 0$, \ie $\lim_{x\to 0^+} xf(x) = 0$.
\end{proof}

\section{General properties of the dynamics}
\label{app:dynamics}

\dynamics*

\begin{proof}
First of all, according to the dynamics described in \cref{sec:dynamics},  $\denmat(\time)$ is obtained as a regularized best-response, \ie a solution of the maximization problem:
\begin{equation}
\label{eq:regbr}
\mirror(\dmat)
	= \argmax\nolimits_{\denmatalt \in \denmats} \braces{\trof{\dmat\denmatalt} - \trof{\hker(\denmatalt)}}.
\end{equation}
Letting $\stateof{\time} = \sum_{\pure=1}^{\vdim} \denval_{\pure}(\time) \, \denvec_{\pure}(\time) \denvec_{\pure}^{\dag}(\time)$ be an eigendecomposition of $\stateof{\time}$, since $\hreg$ is steep, we readily obtain that $\denmat(\time) \in \relint{\denmats}$, which implies that $\denval_{\pure}(\time) > 0$ for all $\pure$. Since $\inf_{x\in(0,1]}\hker''(x) > 0$, the function $x\mapsto\hker(x)$ is strictly convex, and so is $\denmat \mapsto \trof{\hker(\denmat)}$, see \citep[Theorem~2.10]{Carlen2009}. Hence, \eqref{eq:regbr} has a unique solution in $\denmats$. 

By the KKT conditions, the dual variables associated with the positive semi-definiteness constraints, $\denval_{\pure}(\time) \geq 0$ for $\pure = 1,\dots,\vdim$, are equal to zero in the optimal solution, since the inequalities are strict, as argued before. Hence, it is enough to consider the ``reduced'' Lagrangian:
\begin{equation}
\lag(\denmat;\lambda) = \trof{\dmat\denmat} - \trof{\hker(\denmat)} - \lambda(\trof{\denmat}-1)
\end{equation}
where $\lambda \in \R$ is the dual variable associated with the constraint $\trof{\denmat} = 1$. Then, differentiating $\lag$ with respect to $\denmat^\top$, the solutions need to satisfy $\nabla_{\denmat^\top}\lag(\denmat;\lambda) = \dmat - \hker'(\denmat) - \lambda\eye = 0$, \ie
\begin{equation}
\label{eq:sol}
\hker'(\denmat) = \dmat  - \lambda\eye
\end{equation}
where we used that $\nabla_{\denmat^\top}\trof{\hker(\denmat)} = \theta'(\denmat)$, see \citep{PP2012}.
Differentiating \eqref{eq:sol} with respect to $\time$, and invoking that $\dot\dmat = \payfield$, we obtain:
\begin{align}
\ddt\hker'(\denmat) &= \dot\dmat - \dot\lambda\eye \notag\\
&= \payfield -  \dot\lambda\eye
\end{align}
Writing $\theta'(\denmat)$ in the same eigenbasis as $\denmat$, \ie $\hker'(\denmat) = \sum_{\pure = 1}^{\vdim} \hker'(\denval_{\pure})\denvec_{\pure}\denvec_{\pure}^{\dag}$, the above equation becomes:
\begin{align}
\label{eq:time-derivative}
\sum_{k = 1}^{\vdim}\hker''(\denval_k)\dot\denval_{k}\denvec_{k}\denvec_{k}^{\dag} + \sum_{k = 1}^{\vdim}\hker'(\denval_k)\dot\denvec_{k}\denvec_{k}^{\dag} + \sum_{k = 1}^{\vdim}\hker'(\denval_k)\denvec_{k}\dot\denvec_{k}^{\dag} =  \payfield -  \dot\lambda\eye
\end{align}
Applying $\denvec_{\pure}^{\dag}$ on the left and $\denvec_{\purealt}$ on the right of \eqref{eq:time-derivative}, we obtain:
\begin{equation}
\sum_{k = 1}^{\vdim}\hker''(\denval_k)\dot\denval_{k}\denvec_{\pure}^{\dag}\denvec_{k}\denvec_{k}^{\dag}\denvec_{\purealt} + \sum_{k = 1}^{\vdim}\hker'(\denval_k)\denvec_{\pure}^{\dag}\dot\denvec_{k}\denvec_{k}^{\dag}\denvec_{\purealt} + \sum_{k = 1}^{\vdim}\hker'(\denval_k)\denvec_{\pure}^{\dag}\denvec_{k}\dot\denvec_{k}^{\dag}\denvec_{\purealt} =  \denvec_{\pure}^{\dag}\payfield\denvec_{\purealt} -  \dot\lambda\denvec_{\pure}^{\dag}\denvec_{\purealt}
\end{equation}
and using that $\denvec_{k}^{\dag}\denvec_{\ell} = \delta_{k\ell}$, the above relation becomes:
\begin{equation}
\label{eq:fancy-name}
\hker''(\denval_{\pure})\dot\denval_{\pure}\delta_{\pure\purealt} + \hker'(\denval_{\purealt})\denvec_{\pure}^{\dag}\dot\denvec_{\purealt} + \hker'(\denval_{\pure})\dot\denvec_{\pure}^{\dag}\denvec_{\purealt} =  \payfield_{\pure\purealt} -  \dot\lambda\delta_{\pure\purealt}
\end{equation}
Now, we observe that since $\denvec_{\pure}^{\dag}\denvec_{\purealt} = \delta_{\pure\purealt}$, differentiating it with respect to $\time$, we get:
\begin{equation}
\dot\denvec_{\pure}^{\dag}\denvec_{\purealt} + \denvec_{\pure}^{\dag}\dot\denvec_{\purealt} = 0
\end{equation}
and, hence, \eqref{eq:fancy-name} becomes:
\begin{equation}
\label{eq:key}
\hker''(\denval_{\pure})\dot\denval_{\pure}\delta_{\pure\purealt} + (\hker'(\denval_{\purealt}) - \hker'(\denval_{\pure}))\denvec_{\pure}^{\dag}\dot\denvec_{\purealt} =  \payfield_{\pure\purealt} -  \dot\lambda\delta_{\pure\purealt}
\end{equation}
With the above equation in hand, we proceed to the final steps of the proof. Following the same procedure as before, the $\pure\purealt$-entry of $\dot\denmat$, can be written as:
\begin{equation}
\label{eq:dot-expression}
[\dot\denmat]_{\pure\purealt} = \dot\denval_{\pure}\delta_{\pure\purealt} + (\denval_{\purealt} - \denval_{\pure})\denvec_{\pure}^{\dag}\dot\denvec_{\purealt}
\end{equation}

\begin{enumerate}
\item
For $\pure = \purealt$, equation \eqref{eq:dot-expression} gives:
\begin{equation}
\label{eq:dot-diag}
[\dot\denmat]_{\pure\pure} = \dot\denval_{\pure}
\end{equation}
and, equation \eqref{eq:key} becomes:
\begin{equation}
\hker''(\denval_{\pure})\dot\denval_{\pure} =  \payfield_{\pure\pure} - \dot\lambda
\end{equation}
Since $\denval_{\pure} >0$, by the hypothesis on $\hker$, we have $\hker''(\denval_{\pure}) > 0$, and therefore:
\begin{equation}
\label{eq:dot-eval}
\dot\denval_{\pure} =  \frac{\payfield_{\pure\pure}}{\hker''(\denval_{\pure})} - \frac{\dot\lambda}{\hker''(\denval_{\pure})}
\end{equation}
Summing the above for $\purealt = 1,\dots,\vdim$, and using the fact that $\sum_{\purealt = 1}^{\vdim}\dot\denval_{\purealt} = 0$ (since $\sum_{\purealt = 1}^{\vdim}\denval_{\purealt} = 1$), we obtain:
\begin{equation}
\label{eq:dot-lambda}
\dot\lambda = \frac{\sum_{\purealt = 1}^{\vdim}
			\payent_{\purealt\purealt}
			/\hker''(\denval_{\purealt})}
			{\sum_{\purealt = 1}^{\vdim} 1/\hker''(\denval_{\purealt})}
\end{equation}
Hence, combining \eqref{eq:dot-lambda} with \eqref{eq:dot-diag} and \eqref{eq:dot-eval}, we obtain:
\begin{equation}
\label{eq:eigval_dyn}
[\dot\denmat]_{\pure\pure} =  \frac{\payent_{\pure\pure}}{\hker''(\denval_{\pure})}
		- \frac
			{\sum_{\purealt} \payent_{\purealt\purealt}/\hker''(\denval_{\purealt})}
			{\sum_{\purealt} \hker''(\denval_{\pure})/\hker''(\denval_{\purealt})}
\end{equation}

\item
For $\pure = \purealt$, equation \eqref{eq:dot-expression} gives:
\begin{equation}
\label{eq:dot-nondiag}
[\dot\denmat]_{\pure\purealt} = (\denval_{\purealt} - \denval_{\pure})\denvec_{\pure}^{\dag}\dot\denvec_{\purealt}
\end{equation}
and, equation \eqref{eq:key} becomes:
\begin{equation}
\label{eq:eigvec_1}
(\hker'(\denval_{\purealt}) - \hker'(\denval_{\pure}))\denvec_{\pure}^{\dag}\dot\denvec_{\purealt} =  \payfield_{\pure\purealt} 
\end{equation}
Expressing $\dot\denvec_{\purealt}$ in the basis $\denvec_{1},\dots,\denvec_{\vdim}$, we write it as $\dot\denvec_{\purealt} = \sum_{k=1}^{\vdim} B_{\purealt k}\denvec_{k}$ with the coefficients $B_{\purealt k}$'s to be determined. Hence, \eqref{eq:eigvec_1} can be written as:
\begin{equation}
(\hker'(\denval_{\purealt}) - \hker'(\denval_{\pure}))\sum_{k = 1}^{\vdim}B_{\purealt k}\denvec_{\pure}^{\dag}\denvec_{k} =  \payfield_{\pure\purealt} 
\end{equation}
and, since $\denvec_{k}^{\dag}\denvec_{\ell} = \delta_{k\ell}$, we readily get:
\begin{equation}
(\hker'(\denval_{\purealt}) - \hker'(\denval_{\pure}))B_{\purealt\pure} =  \payfield_{\pure\purealt} 
\end{equation}
or, equivalently:
\begin{equation}
\label{eq:Betas}
B_{\purealt\pure} =  \frac{\payfield_{\pure\purealt}}{\hker'(\denval_{\purealt}) - \hker'(\denval_{\pure})}
\end{equation}
Now, it is easy to see that $\denvec_{\pure}^{\dag}\dot\denvec_{\purealt} = B_{\purealt\pure}$. Therefore, combining it with \eqref{eq:dot-nondiag} and \eqref{eq:Betas}, we conclude that:
\begin{align}
\label{eq:eigvec_dyn}
[\dot\denmat]_{\pure\purealt} = \frac{\denval_{\purealt} - \denval_{\pure}}{\hker'(\denval_{\purealt}) - \hker'(\denval_{\pure})}\payfield_{\pure\purealt}
\end{align}
\end{enumerate}
This concludes the proof.
\end{proof}

\begin{proposition}
\label{prop:basis_free_dynamics}
The \ac{QRD} can be written in the form:
\begin{equation}
\label{eq:QRD-free_prop}
\dot\denmat
	= \int_{0}^{1} \denmat^{1-s} \payfield(\denmat) \denmat^{s} \dd s
		- \trof{\denmat \payfield(\denmat)} \denmat
\end{equation}
\end{proposition}
\begin{proof}
In the case of the von Neumann regularizer, we ontain the replicator dynamics:
\begin{equation}
\dot\dmat
	= \payfield(\denmat)
	\qquad
\denmat
	= \frac{\exp(\dmat)}{\trof{\exp(\dmat)}}
\end{equation}
Hence, differentiating with respect to $\time$ and using the chain rule, we obtain:
\begin{align}
\label{eq:basis_free}
\dot\denmat
	&= \frac{1}{\trof{\exp(\dmat)}}\ddt\exp\parens{\dmat} + \exp(\dmat)\ddt\frac{1}{\trof{\exp(\dmat)}} \notag\\
	&= \frac{1}{\trof{\exp(\dmat)}}\ddt\exp\parens{\dmat} - \frac{\exp(\dmat)}{\trof{\exp(\dmat)}^2}\ddt\trof{\exp(\dmat)} \notag\\
	&= \frac{1}{\trof{\exp(\dmat)}}\ddt\exp\parens{\dmat} - \frac{\exp(\dmat)}{\trof{\exp(\dmat)}^2}\trof*{\ddt\exp(\dmat)}
\end{align}
By the Fr\'echet derivative \cite{Wil67}, it holds:
\begin{align}
\label{eq:frechet}
	\ddt\exp\parens{\dmat} &= \int_{0}^{1}e^{(1-s)\dmat}\dot\dmat e^{s\dmat}\dd s \notag\\
	&= \int_{0}^{1}e^{(1-s)\dmat}\payfield(\denmat) e^{s\dmat}\dd s \notag\\
	&= \trof{\exp(\dmat)}\int_{0}^{1}\denmat^{1-s}\payfield(\denmat) \denmat^{s}\dd s
\end{align}
and, using \eqref{eq:frechet} in \eqref{eq:basis_free}, we get:
\begin{align}
\dot\denmat
	&= \frac{1}{\trof{\exp(\dmat)}}\trof{\exp(\dmat)}\int_{0}^{1}\denmat^{1-s}\payfield(\denmat) \denmat^{s}\dd s
	\notag\\
	&\hspace{4em}
		- \frac{\exp(\dmat)}{\trof{\exp(\dmat)}^2}\trof*{\trof{\exp(\dmat)}\int_{0}^{1}\denmat^{1-s}\payfield(\denmat) \denmat^{s}\dd s}
	\notag\\
	&= \int_{0}^{1}\denmat^{1-s}\payfield(\denmat) \denmat^{s}\dd s - \frac{\exp{\dmat}}{\trof{\exp(\dmat)}} \trof*{\int_{0}^{1}\denmat^{1-s}\payfield(\denmat) \denmat^{s}\dd s}
	\notag\\
	&= \int_{0}^{1}\denmat^{1-s}\payfield(\denmat) \denmat^{s}\dd s - \frac{\exp{\dmat}}{\trof{\exp(\dmat)}} \int_{0}^{1}\trof{\denmat^{1-s}\payfield(\denmat) \denmat^{s}}\dd s
	\notag\\
	&= \int_{0}^{1}\denmat^{1-s}\payfield(\denmat) \denmat^{s}\dd s - \frac{\exp{\dmat}}{\trof{\exp(\dmat)}} \int_{0}^{1}\trof{\denmat\payfield(\denmat)}\dd s
	\notag\\
	&= \int_{0}^{1}\denmat^{1-s}\payfield(\denmat) \denmat^{s}\dd s - \trof{\denmat\payfield(\denmat)}\denmat
\end{align}
as asserted.
\end{proof}

\section{Regret minimization}
\label{app:regret}

Our aim in this appendix is to prove that the dynamics \eqref{eq:FTQL} incur at most constant regret.
For convenience, we restate the relevant result below.

\regret*

Our proof essentially follows the analysis of \cite{KM17};
the specialization of the proof occurs in deriving the exact value of the regret bound \eqref{eq:reg-FTQL} but, for completeness, we provide the entire proof.

\begin{proof}
Let $\bench_\play \in \denmats$ be a best fixed action in hindsight, \ie:
\begin{equation}
\label{eq:best_fixed}
\bench_\play \in \argmax_{\denmatalt_{\play}\in\denmats_{\play}}
\int_{0}^{\horizon}\pay_{\play}(\denmatalt_{\play};\stateof[-\play]{\time})\dd \time    
\end{equation}
and define the function $\energy_\play(\time) \defeq \fench_{\play}(\bench_\play,\dmat_\play(\time))$ for $\time \geq 0$.
By \cref{lem:Fench}, we have that:
\begin{align}
\label{eq:fencheldiff}
\ddt\energy_\play(\time)
	= \trof{\payfield_\play(\denmat(\time))(\denmat_\play(\time)-\bench_\play)}
	= \pay_{\play}(\denmat(\time)) - \pay_{\play}(\bench_{\play};\denmat_{-\play}(\time))
\end{align}
Integrating over time, and combining \eqref{eq:regret} with \eqref{eq:best_fixed}, we obtain:
\begin{align}
\label{eq:regbound}
\reg_\play(\horizon)
    &= \energy_\play(0) - \energy_\play(\horizon)
    \notag\\
    &= \hconj_{\play}(\dmat_{\play}(0))
		- \trof{\bench_{\play}\dmat_{\play}(0)} - \hconj_{\play}(\dmat_{\play}(\horizon))
		+ \trof{\bench_{\play}\dmat_{\play}(\horizon)}
    \notag\\
    &\leq \hconj_{\play}(\dmat_{\play}(0))
		- \trof{\bench_{\play}\dmat_{\play}(0)} + \hreg_{\play}(\bench_{\play})
    \notag\\
    &= \hreg_{\play}(\bench_{\play}) - \hreg_{\play}(\mirror_{\play}(\dmat_\play(0)))
    \notag\\
    &= \trof{\hker_{\play}(\bench_{\play})} - \trof{\hker_{\play}(\mirror_{\play}(\dmat_\play(0)))}
    \notag\\
    &\leq \max_{\denmatalt_{\play}\in\denmats_{\play}} \trof{\hker_{\play}(\denmatalt_{\play})} - \min_{\denmatalt_{\play}\in\denmats_{\play}} \trof{\hker_{\play}(\denmatalt_{\play})}
\end{align}
where we used the fact that
\begin{align}
\hconj_{\play}(\dmat_{\play})
	= \max\nolimits_{\denmatalt_{\play}\in\denmats_{\play}}\setof{\trof{\denmatalt_{\play}\dmat_{\play}} - \hreg_{\play}(\denmatalt_{\play})}
    \geq \trof{\bench_{\play}\dmat_{\play}} - \hreg_{\play}(\bench_{\play}) 
\end{align}
with equality if and only if $\denmatalt_{\play} = \mirror_\play(\dmat_{\play})$.
The precise bound \eqref{eq:reg-FTQL} then follows by noting that the maximum difference in the values of $\hreg$ is attained between the barycenter of the spectraplex and any pure density matrix of rank $1$, which thus hields the bound
\begin{equation}
\max\hreg_{\play} - \min\hreg_{\play}
	= \abs{\vdim_{\play} \cdot \hker_{\play}(1/\vdim_{\play}) - \hker_{\play}(1)},
\end{equation}
and our assertion follows.
\end{proof}

\section{The folk theorem}
\label{app:folk}

\newmacro{\pcpt}{\boldsymbol{\mathcal{C}}}

In this appendix, our aim is to prove \cref{thm:folk}.
To lay the necessary groundwork, we begin with a series of helper lemmas.
The first is related to the verification parts of \cref{thm:folk}, and concerns the behavior of the dynamics near its trapping regions.

\begin{lemma}
\label{lem:trap}
Fix some $\eq\in\denmats$, and assume that every neighborhood $\pnhd$ of $\eq$ in $\denmats$ admits a solution trajectory $\denmat(\time) = \mirror(\dmat(\time))$ of \eqref{eq:FTQL} such that $\denmat(\time) \in \pnhd$ for all $\time \geq 0$.
Then $\eq$ is a Nash equilibrium.
\end{lemma}

\begin{proof}
We argue by contradiction.
Indeed, if $\eq$ is not Nash, then, by the variational characterization \eqref{eq:VI} of \aclp{NE}, there exists some deviation $\basemat\in\denmats$ such that $\trof{\payfield(\eq)\,(\basemat - \eq)} > 0$.
Then, by continuity, there exists a sufficiently small compact neighborhood $\pcpt$ of $\eq$ in $\denmats$ such that $\trof{\payfield(\denmat) \, (\basemat - \denmat)} > 0$ for all $\denmat \in \pcpt$, implying in turn that $\const \defeq \inf_{\denmat\in\pcpt} \trof{\payfield(\denmat) \, (\basemat - \denmat)} > 0$, 
Now, by assumption, there exists a solution orbit $\stateof{\time} = \mirror(\dstateof{\time})$ of \eqref{eq:FTQL} such that $\stateof{\time} \in \pnhd$ for all $\time\geq0$ so, by \cref{lem:Fench}, we readily obtain
\begin{equation}
\ddt \fench(\basemat,\dstateof{\time})
	= \trof{\payfield(\stateof{\time}) \, (\stateof{\time} - \basemat)}
	\leq -\const
	< 0
	\quad
	\text{for all $\time\geq0$}.
\end{equation}
A simple integration then yields $\fench(\basemat,\dstateof{\time}) \to -\infty$ as $\time\to\infty$, a contradiction that proves our claim.
\end{proof}

The second intermediate result we will need is a uniform convergence lemma for states that are consistently asymptotically stable:

\begin{lemma}
\label{lem:uniform}
Suppose that $\eq \in \relint\denmats$ is consistently asymptotically stable under \eqref{eq:FTQL}.
Then, for every sufficiently small compact neighborhood $\pnhd$ of $\eq$ in $\relint\denmats$, we have
\begin{equation}
\label{eq:uniform}
\lim_{\time \to \infty} \sup_{\denmat \in \pnhd} \norm{\pflowof{\time}{\denmat}-\eq}
	= 0.
\end{equation}
\end{lemma}

\begin{proof}
First of all, since $\eq \in \relint\denmats$, we have that $\ker(\eq) = \setof{0}$, which implies that $ \denmats_{\eq} = \relint\denmats$.
Hence, for any open set $O$ in $\relint\denmats$, it holds that $O\cap\denmats_{\eq} = O$.
Now, let $\pnhd_0$ be the basin of attraction of $\eq$, according to the definition of the consistent asymptotic stability in \cref{def:consistent}, and let $\pnhd \subseteq \pnhd_0$ be a compact neighborhood of $\eq$.
Suppose, for the sake of contradiction, that $\sup_{\denmat \in \pnhd} \norm{\pflowof{\time}{\denmat}-\eq} \not\to 0$.
This implies that there exists $\eps >0$, a sequence $\{\time_n\}_{n \in \N}$ with $\time_n \to \infty$ as $n \to \infty$, and $\denmat_{n} \in \pnhd$ for $n \in \N$, such that
\begin{equation}
\label{eq:largedist}
\norm{\pflowof{\time_n}{\denmat_n}-\eq} \geq \eps \quad \text{for all } n \in \N
\end{equation}
Since $\pnhd$ is compact, we may assume (by taking subsequences, if necessary) that $\denmat_n$ converges to some limit point $\limitden \in \pnhd$.
Now, since $\eq$ is Lyapunov stable, there exists a neighborhood $\pnhd'$ of $\eq$, such that the trajectory $\pflowof{\time}{\denmat}$ remains within $\eps/2$-distance of $\eq$, if $\denmat \in \pnhd'$, or, equivalently:
\begin{equation}
\norm{\pflowof{\time}{\denmat}-\eq} < \frac{\eps}{2} \quad \text{for all } \time \geq 0 \text{ and } \denmat \in \pnhd'
\end{equation}
Define the hitting time $\hit \defeq \inf\{\time \geq 0: \pflowof{\time}{\limitden} \in \pnhd'\}$, as the first time that the trajectory enters $\pnhd'$ when starting from $\limitden$.
It is easy to see that $\hit < \infty$, since $\limitden \in \pnhd \subseteq \pnhd_0$, and, thus, $\norm{\pflowof{\time}{\limitden}-\eq} \to 0$ as $\time \to \infty$ by the asymptotic stability of $\eq$.

By continuity of $\pflow$, we readily get that there exists a neighborhood $\cont$ of $\limitden$ such that $\pflowof{\hit}{\cont} \subseteq \nhd''$, where $\nhd''$ is a neighborhood of $\eq$ with $\norm{\pflowof{\time}{\denmat}-\eq} < \eps$ for all $\time \geq 0$ and $\denmat \in \nhd''$.

Since $\denmat_n \to \limitden$, we conclude that $\denmat_n \in \cont$ for all sufficiently large $n$, which, in turn, implies that $\pflowof{\hit}{\denmat_n} \in \nhd''$.
Moreover, by definition of the sequence $\{\time_n\}_{n \in \N}$, we have that $\time_n \to \infty$, which gives that $\time_n > \hit$ for $n$ sufficiently large, since $\hit < \infty$, as argued before.

Therefore, for all $n$ sufficiently large, we have $\norm{\pflowof{\time_n}{\denmat_n}-\eq} < \eps$, which contradicts \eqref{eq:largedist}.
\end{proof} 

Our last step before proving \cref{thm:folk} is a comparison between the notions of Lyapunov stability induced by the flow $\dflow$ on $\denmats$ via $\mirror$, and the corresponding notion relative to the conjugate flow $\pflow$ on $\denmats$.
Formally, we have the following definition.

\begin{definition}
We say that $\eq \in \denmats$ is:
\begin{enumerate}
\item
$\pflow$-stable, if for all neighborhoods $\pnhd$ of $\eq \in \denmats$, there exists a neighborhood $\pnhd'$ of $\eq \in \denmats$, such that $\pflowof{\time}{\denmat} \in \pnhd$ for all $\denmat \in \pnhd'$ and all $\time \geq 0$.
\item
$\dflow$-stable, if for all neighborhoods $\pnhd$ of $\eq \in \denmats$, there exists a neighborhood $\pnhd'$ of $\eq \in \denmats$, such that $\pflowof{\time}{\denmat} \in \pnhd$ for all $\denmat \in \pnhd' \cap \im\mirror$ and all $\time \geq 0$.
\end{enumerate}
\end{definition}

\begin{lemma} 
\label{lem:stable-conj}
A point $\eq \in \denmats$ is $\pflow$-stable if and only if it is $\dflow$-stable.
\end{lemma}

\begin{proof}
The ``only if'' part is trivial, so we focus on the ``if'' part.
To that end, suppose that $\eq$ is $\dflow$-stable but not $\pflow$-stable.
Then, there exists a neighborhood $\pnhd$ of $\eq$ in $\denmats$ such that for all neighborhoods $\pnhd'$ of $\eq$, there exists $\denmat \in \pnhd'$ such that $\pflowof{\time}{\denmat}\not\in \nhd$ for some $\time \geq 0$.
Now, let $\pnhd_0$ be a neighborhood of $\eq$ such that $\cl(\pnhd_0) \subseteq \pnhd$.
Then, by $\dflow$-stability of $\eq$, there exists a neighborhood $\pnhd_0'$ of $\eq$ such that $\pflowof{\time}{\denmat} \in \pnhd_0$ for all $\denmat \in \pnhd_0'\cap\im\mirror$ and all $\time \geq 0$.
But, since $\eq$ is not $\pflow$-stable, there exists $\denmat \in \pnhd_0'\cap(\im\mirror)^c$ such that the hitting time $\hit = \inf \setof{\time \geq 0: \pflowof{\time}{\denmat}\not\in \pnhd}$ is finite.
Moreover, let $\setof{\denmat_n}_{n \in \N} \subseteq \pnhd_0'\cap\im\mirror$ be a sequence converging to $\denmat$ (such a sequence exists, because $\cl(\im\mirror) = \denmats$).
Since $\eq$ is $\dflow$-stable, it holds that $\pflowof{\hit}{\denmat_n} \in \pnhd_0$ for all $n$, and, therefore, using the continuity of $\pflow$, we get:
\begin{equation}
\pflowof{\hit}{\denmat} = \lim_{n \to \infty}  \pflowof{\hit}{\denmat_n}
\end{equation}
which implies that $\pflowof{\hit}{\denmat} \in \cl(\pnhd_0) \subseteq \pnhd$, as a limit point of $\pnhd_0$.
But, since $\pflowof{\hit}{\denmat} \not\in\pnhd$ by assumption, we obtain a contradiction that completes our proof.
\end{proof}

Thanks to the above proposition, we do not need to distinguish between ``Lyapunov stable'' and ``consistently Lyapunov stable'' states, as the two notions are equivalent.
Thus, with this issue settled, we are in a position to prove our quantum version of the ``folk theorem'', which we restate for convenience below:

\folk*

We prove the various parts of the theorem below.

\begin{proof}[Proof of Part 1]
Suppose that $\eq$ is a \acl{NE} of $\qgame$.
If $\rank(\eq) = 1$, there is nothing to show, so it suffices to establish our claim for the case $\rank(\eq) > 1$.
In that case, by restricting the dynamics if necessary to the relative interior of the face $\eqs = \setdef{\denmat\in\denmats}{\ker\denmat \geq \ker\eq}$ of $\denmats$ whose relative interior contains $\eq$, we may assume without loss of generality that $\eq$ has, in fact, full rank.
With this in mind, fix some $\dmat \in \mirror^{-1}(\eq)$, and consider the trajectory $\dstateof{\time} = \dmat + \time \payfield(\eq)$ starting at $\dmat$ at time $\time = \tstart$.
Since $\payfield(\eq) \in \ncone(\eq)$ by the variational characterization \eqref{eq:VI} of \aclp{NE}, we conclude by \cref{lem:mirror} that $\mirror(\dstateof{\time}) = \mirror(\dmat)$ for all $\time\geq\tstart$.
In turn, this shows that $\dot\dmat(\time) = \payfield(\eq) = \payfield(\mirror(\dstateof{\time}))$ for all $\time = \tstart$, \ie $\dstateof{\time}$ is a solution orbit of \eqref{eq:FTQL}.
By the Picard-Lindelöf theorem, this implies that $\dstateof{\time}$ is the unique solution thereof;
hence, given that $\stateof{\time} = \mirror(\dstateof{\time}) = \eq$ for all $\time\geq\tstart$, our claim follows.
\end{proof}

\begin{proof}[Proof of Parts 2 and 3]
Convergence to and Lyapunov stability of $\eq$ both imply the assumption of \cref{lem:trap}, so our assertions follow immediately by the conclusion of said lemma.
\end{proof}

\begin{proof}[Proof of Part 4]
The proof of this part consists of three steps.
First, we will construct a measure $\invar$ on $\relint\denmats$ that is invariant under the flow $\pflow$, \ie $\invar(A) = \invar(\pflowof{\time}{A})$ for all measurable sets $A$ and $\time \geq 0$, where by $\pflowof{\time}{A}$ we denote the image of $A$ after time $\time$.
Next, we will show that if $\eq \in \relint \denmats$, it cannot be consistently asymptotically stable.
Finally, we will conclude that if $\eq \in \denmats$, is consistently asymptotically stable, then it is pure with a simple reduction argument.

\para{Step 1} By the definition of $\mirror$, we have $\dmat = \nabla_{\denmat^{\top}}\hker(\denmat) + c\cdot\eye$, \ie $\mirror^{-1}(\denmat) = \{\nabla_{\denmat^{\top}}\hker(\denmat) + c\cdot\eye : c \in \R\}$.
Let $\quotientspace \defeq\dmats/\sim$ be the quotient space induced by the equivalence relation $\sim$, where $\dmat \sim \dmatalt$ if $\dmat - \dmatalt = c\cdot\eye$ for some $c \in \R$, and let $\quotientmap: \dmats \to \quotientspace$ be the corresponding quotient map, defined as $\quotientmap(\dmat) = \{\dmat + c\cdot\eye: c \in \R \}$.
Thus, by descending to the quotient space $\quotientspace$, the mirror map $\mirror$ factors through $\quotientspace$ as $\quotientmirror: \quotientspace \to \denmats$ defined as $\quotientmirror \defeq \mirror \circ \quotientmap$.
Now, let $\Leb$ denote the Lebesgue measure on $\quotientspace$, and let $\invar = \quotientmirror \texttt{\#} \Leb$ denote the pushforward of $\Leb$ to $\denmats$ via $\quotientmirror$, \ie $\invar(A) = \Leb(\quotientmirror^{-1}(A))$ for all measurable sets $A$ in $\denmats$.
Finally, by \cref{prop:volume-0}, we readily get that $\invar(\pflowof{\time}{A}) = \invar(A)$ for all measurable sets $A \subseteq \denmats$ and all $\time \geq 0$.

\para{Step 2} For the sake of contradiction, suppose $\eq \in \relint\denmats$ is consistently asymptotically stable, and let $\pnhd$ be a sufficiently small compact neighborhood of $\eq$, as per \cref{lem:uniform}.
Then, it holds $\lim_{\time \to \infty} \sup_{\denmat \in \pnhd} \norm{\pflowof{\time}{\denmat}-\eq} = 0$, which implies that:
\begin{equation}
\label{eq:measure_shrink}
\lim_{\time\to\infty} \invar(\pflowof{\time}{\pnhd})
    = \invar(\setof{\eq})
    = 0
    < \invar(\pnhd).
\end{equation}
However, by \emph{Step 1} of the proof, we have that $\lim_{t \to \infty}\invar(\pflowof{\time}{\pnhd}) = \invar(\pnhd)$, since $\invar(\pflowof{\time}{\pnhd}) = \invar(\pnhd)$ for all $\time \geq 0$.
Hence, combining it with \eqref{eq:measure_shrink}, we arrive at a contradiction.

\para{Step 3} Suppose, finally, that $\eq$ is consistently asymptotically stable with $\rankeq\defeq\rank(\eq) > 1$.
Note that if $\eq$ is full-rank, it cannot be consistently asymptotically stable, as shown at the \emph{Step 2} of the proof.
Let $\A_{\eq} \defeq \setdef{\denmat \in \denmats}{\ker(\denmat) \geq \ker(\eq)}$, which is convex.
Indeed, let $\denmat, \denmatalt \in \A_{\eq}$ and $\lambda \in (0,1)$.
For $y \in \ker(\eq)$, we have:
\begin{align}
[\lambda \denmat + (1-\lambda)\denmatalt]y & = \lambda \denmat y+ (1-\lambda)\denmatalt y \notag\\
& = 0
\end{align}
\ie $y \in \ker(\lambda \denmat + (1-\lambda)\denmatalt)$, and, thus, $\ker(\lambda \denmat + (1-\lambda)\denmatalt) \geq \ker(\eq)$.

Hence, considering the restriction of the dynamics on $\A_{\eq}$, \emph{Step 2} of the proof shows that a point in $\relint\A_{\eq}$ cannot be consistently asymptotically stable on the induced topology. Therefore, the restriction of $\eq$ on $\A_{\eq}$, denoted $\eq\mid_{\A_{\eq}}$ cannot be consistently asymptotically stable on $\A_{\eq}$, since it has full-rank on $\A_{\eq}$.
So, it remains to show that $\eq$ cannot be consistently asymptotically stable on $\denmats$.

For the sake of contradiction, suppose it is.
Then, according to \cref{def:consistent}, there exist $\pnhd$ of $\eq$ in $\denmats$ such that $\lim_{\time\to\infty} \pflowof{\time}{\denmat} = \eq$ for all $\denmat\in\pnhd\cap\denmats_{\eq}$.
But, since $\pnhd$ is a neighborhood of $\eq$ in $\denmats$, \begin{equation}
    \parens*{\pnhd\cap\denmats_{\eq}}\cap\A_{\eq} = \pnhd\cap\setof{\denmatalt \in \denmats: \ker(\denmatalt) = \ker(\eq)} \supsetneq \setof{\eq}
\end{equation}
and, thus, the restriction of $\pnhd\cap\setof{\denmatalt \in \denmats: \ker(\denmatalt) = \ker(\eq)}$ on $\A_{\eq}$ is a neighborhood of $\eq$ in $\A_{\eq}$.
Therefore, by our previous argument on the induced dynamics on $\A_{\eq}$, there exists $\denmat_0 \in \pnhd\cap\setof{\denmatalt \in \denmats: \ker(\denmatalt) = \ker(\eq)}$ such that $\pflowof{\time}{\denmat_0} \not\to \eq$, as $\time \to \infty$.
This contradicts our original assumption, so our proof is complete.
%
\end{proof}

\begin{proof}[Proof of Part 5]
Suppose that $\eq$ is (locally) variationally
stable, \ie $\trof{\payfield(\denmat)(\denmat -\eq)} < 0$ for all $\denmat \in \pnhd$, where $\pnhd$ is a neighborhood of $\eq$.
Then, by Lyapunov's direct method for the energy function $\energy(\time) \defeq \fench(\eq,\dmat(\time))$, and invoking \cref{lem:stable-conj}, we get that $\eq$ is stable, thus, there exists neighborhood $\pnhd'$ of $\eq$, such that $\denmat(\time) \in \pnhd$ for all $\time\geq 0$, if $\denmat(0) \in \pnhd'$.
This means that if $\denmat(0) \in \pnhd'$, then $\trof{\payfield(\denmat(\time))(\denmat(\time) -\eq)} < 0$ for all $\time \geq 0$.
In what follows, we will show that $\denmat(\time) \to \eq$ if $\denmat(0) \in \pnhd'$.

Now, let $\denmat(0) \in \pnhd'$.
Since $\denmat(\time) \in \pnhd$ for all $\time \geq 0$, and using \cref{lem:Fench}, we have that:
\begin{equation}
    \ddt \energy(\time) = \trof{\payfield(\denmat(\time))(\denmat(\time)-\eq)} < 0
\end{equation}
where the last inequality comes from \eqref{eq:VS}.
Hence, $\energy(\time)$ is decreasing if $\denmat(0) \in \pnhd$.

Now, we will show that $\eq$ is an $\omega$-limit of the flow, \ie there exists a sequence $\setof{\time_n}_{n \in \N}$ such that $\denmat(\time_n) \to \eq$ as $n \to \infty$.
For the sake of contradiction, suppose such a sequence does not exist.
Then, $\norm{\denmat(\time) - \eq}$ is bounded away from zero, which implies that there exists $c >0$ such that $\trof{\payfield(\denmat(\time))(\denmat(\time)-\eq)} < -c$, \ie $\ddt \energy(\time) < -c$.
Integrating over time, we get
\begin{align}
    \energy(\time) - \energy(0) < -c\time \notag
\end{align}
which implies that $\energy(\time)\to -\infty$ as $\time \to \infty$.
This contradicts the non-negativity property of the Fenchel coupling, thus, we conclude that there exists a sequence $\setof{\time_n}_{n \in \N}$ such that $\denmat(\time_n) \to \eq$ as $n \to \infty$.

Finally, by \cref{prop:reciprocity}, we get that $\fench(\eq,\dmat(\time_n)) \to 0$ as $n\to\infty$, \ie $\energy(\time_n) \to 0$.
Therefore, since $\energy(\time)$ is decreasing and $\energy(\time_n) \to 0$, we conclude that $\energy(\time) \to 0$ as $\time\to\infty$, \ie $\fench(\eq,\dmat(\time)) \to 0$.
Again, using \cref{prop:reciprocity}, we conclude that $\denmat(\time) \to \eq$, \ie $\eq$ is consistently asymptotically stable.
\end{proof}

\section{Poincaré recurrence}
\label{app:recurrence}

\newmacro{\qmat}{\mathbf{Z}}		
\newmacro{\qmats}{\boldsymbol{\mathcal{Z}}}		
\newmacro{\qval}{z}		

\newmacro{\qnhd}{\boldsymbol{\mathcal{W}}}		
\newmacro{\qmap}{\boldsymbol{\Pi}}		

\newmacro{\qflow}{\boldsymbol{\zeta}}		
\DeclarePairedDelimiterXPP{\qflowof}[2]{\qflow_{#1}}{(}{)}{}{#2}		

\newmacro{\set}{\mathcal{S}}

Our aim in this appendix will be to prove our recurrence result for min-max games admitting a full-rank equilibrium.
For ease of reference, we restate the result below:

\recurrence*

As we briefly discussed in the main part of the paper, our proof strategy will consist of the following steps:
\begin{enumerate}
\item
First, we will need to collapse the space of dual space $\dmats$ of $\vecspace \equiv \herm[\vdim]$ to a subspace $\qmats$ with constant trace.
The reason for this is that the mirror map $\mirror\from\dmats\to\denmats$ is not injective, even when $\hreg$ is steep:
for all $\coef\in\R$, we have $\mirror(\dmat + \coef\eye) = \mirror(\dmat)$, which means that the inverse image of any point in $\im\mirror$ always contains a copy of the real line.
This collapse is intended to ``quotient out'' this redundancy so as to enable volume comparisons later on.
\item
The second step is to show that the dynamics \eqref{eq:FTQL} remain well-defined on the quotient space $\qmats$, and that they are incompressible, that is, the volume of a set of initial conditions is maintained over time.
This is encoded by \cref{prop:volume} in the main text.
\item
The next step in our proof will be to show that the induced dynamics on $\qmats$ admit a \emph{constant of motion}.
This will be provided by the Fenchel coupling relative to the game's interior equilibrium (and it is also the point where the full-rank assumption is needed).
\item
Finally, the last \textendash\ and most challenging \textendash\ part of the proof is to show that the level sets of the Fenchel coupling are bounded in $\qmats$ (by contrast, they are \emph{not} bounded in $\dmats$).
\end{enumerate}
Once these steps our complete, \cref{thm:recurrence} will follow from Poincaré's theorem.
In what follows, we encode each of these steps in a series of lemmas and proposition.
To lighten notation, we will suppress player indices throughout unless it is absolutely necessary to include them.

\para{Step 1: Descending to the quotient}

To quotient out the redundancy mentioned above, consider the transformed matrix variables
\begin{equation}
\label{eq:qmat}
\qmat
	= \dmat - (1/\vdim) \trof{\dmat} \eye
\end{equation}
so $\tr\qmat = 0$ for all $\dmat\in\dmats$.
Formally, this transformation can be represented via tha map $\qmap\from\dmats\to\qmats$, where
\begin{equation}
\label{eq:qmats}
\qmats
	= \setdef{\qmat\in\dmats}{\tr\qmat = 0}
\end{equation}
and $\qmap \from \dmat \mapsto \qmap(\dmat) = \qmat$ is defined via \eqref{eq:qmat} above.
As such, $\qmat$ can be seen as a representative of the equivalence relation $\dmat \sim \dmat + \coef\eye$, $\coef\in\R$, under which light $\qmats$ can be identified with the quotient $\dmats/\sim$.

By construction and the properties of $\mirror$ (\cref{lem:mirror} in particular), we have $\mirror(\dmat) = \mirror(\qmap(\dmat)) = \mirror(\qmat)$ for all $\dmat\in\dmats$.
Accordingly, \eqref{eq:FTQL} readily gives
\begin{align}
\dot\qmat
	= \ddt \bracks{\dmat - (1/\vdim) \trof{\dmat} \eye}
	&= \dot\dmat - (1/\vdim) \trof{\dot\dmat} \eye
	\notag\\
	&= \payfield(\mirror(\dmat)) - (1/\vdim) \trof{\payfield(\mirror(\dmat))} \eye
	\notag\\
	&= \payfield(\mirror(\qmat)) - (1/\vdim) \trof{\payfield(\mirror(\qmat))} \eye
\end{align}
so we can rewrite \eqref{eq:FTQL} in terms of $\qmat$ as
\begin{equation}
\label{eq:FTQL-0}
\tag{FTQL$_{\qmats}$}
\dot\qmat
	= \payfield(\denmat)
		- (1/\vdim) \trof{\payfield(\denmat)} \eye
	\qquad
\denmat
	= \mirror(\qmat).
\end{equation}
Since $\trof{\dot\qmat} = \trof{\payfield(\denmat)} - (1/\vdim) \trof{\payfield(\denmat)} \tr\eye = 0$, it follows that any traceless initial condition of \eqref{eq:FTQL-0} will remain traceless \textendash\ and hence remain in $\denmats$ for all $\run\geq0$.
We thus conclude that \eqref{eq:FTQL-0} is a well-posed dynamical system on $\qmats$ whose induced flow $\qflow\from\R\times\qmats \to \qmats$ conjugates $\dflow$ in the sense that $\qflow_{\time} \circ \qmap = \qmap \circ \dflow_{\time}$ for all $\time\geq\tstart$.

\para{Step 2: Volume preservation}

Our second step is to show that the dynamics \eqref{eq:FTQL-0} are volume-preserving.
Formally, we establish the following direct analogue of \cref{prop:volume}:

\begin{proposition}
\label{prop:volume-0}
Let $\qnhd \subseteq \qmats$ be an open set of initial conditions of \eqref{eq:FTQL-0}, and let $\qnhd_{\time} = \qflowof{\time}{\qnhd}$, $\time\geq0$, denote the evolution of $\qnhd$ under the flow $\qflow \from \R\times\qmats \to \qmats$ of \eqref{eq:FTQL-0}.
Then, $\vol(\dnhd_{\time}) = \vol(\dnhd)$ for all $\time\geq0$.
\end{proposition}

\begin{remark*}
In the above (as in \cref{prop:volume}), the volume form $\vol(\cdot)$ stands for the ordinary Lebesgue measure on $\qmats$.
\end{remark*}

\begin{proof}[Proof of \cref{prop:volume-0}]
The key observation here is that the dynamics \eqref{eq:FTQL-0} are \emph{incompressible} \textendash\ that is, \emph{divergence-free}.
Indeed, since $\payfield_{\play}$ does not depend on $\denmat_{\play}$ (it is not possible to suppress player indices here), we readily have $\nabla_{\qmat_{\play}^{\top}} \payfield_{\play}(\mirror(\dmat)) = 0$ for all $\play\in\players$.
This immediately implies that the field
\begin{equation}
\mathbf{W}_{\play}(\qmat)
	= \payfield_{\play}(\mirror(\qmat))	- (1/\vdim_{\play}) \trof{\payfield_{\play}(\mirror(\qmat))} \eye
\end{equation}
has $\mathrm{div}_{\qmat}(\mathbf{W}(\qmat)) = 0$, so \eqref{eq:FTQL-0} is incompressible.
Our claim then follows from Liouville's formula \cite{Arn89}.
\end{proof}

\cref{prop:volume} is proved in exactly the same way, so we omit its proof (which is not needed for the sequel, anyway).

\para{Step 3: The Fenchel coupling as a constant of motion}

We now proceed to show that the dynamics \eqref{eq:FTQL-0} admit a constant of motion, namely the Fenchel coupling relative to the game's full-rank equilibrium.
Indeed, letting $\energy(\time) = \fench(\eq,\qmat(\time))$, \cref{lem:Fench} readily gives:
\begin{align}
\label{eq:invariance}
\frac{d\energy}{d\time}
	&= \trof{\payfield_{1}(\denmat) (\denmat_{1} - \eq_{1})}
		+ \trof{\payfield_{2}(\denmat) (\denmat_{2} - \eq_{2})}
	\notag\\
	&= \pay_{1}(\denmat_{1},\denmat_{2}) - \pay_{1}(\eq_{1},\denmat_{2})
		+ \pay_{2}(\denmat_{1},\denmat_{2}) - \pay_{2}(\denmat_{1},\eq_{2})
	\notag\\
	&= -\minmax(\denmat_{1},\denmat_{2}) + \minmax(\eq_{1},\denmat_{2})
		+ \minmax(\denmat_{1},\denmat_{2}) - \minmax(\denmat_{1},\eq_{2})
	= \minmax(\eq_{1},\denmat_{2}) - \minmax(\denmat_{1},\eq_{2})
	= 0
\end{align}
where, in the last line, we used the assumption that $\eq$ is a full-rank equilibrium;
see also the relevant discussion in \cite{JW09,JPS22}.

\para{Step 4: Boundedness of the trajectories}

We are left to show that the trajectories of \eqref{eq:FTQL-0} are bounded.
The key to this will be the following technical lemma:

\begin{lemma}
\label{lem:bounded}
Let $\curr[\qmat]$, $\run=\running$, be a sequence in $\qmats$ such that $\fench(\basemat,\curr[\qmat])$ remains constant for some interior $\basemat\in\denmats$.
Then the eigenvalues $\{\qval_{\coord,\run}\}_{\coord=1,\dotsc,\vdim}$ of $\curr[\qmat]$ are bounded, \ie $\sup_{\coord,\run} \abs{\qval_{\coord,\run}} < \infty$.
\end{lemma}

\begin{proof}
We argue by contradiction.
Indeed, assume that $\fench(\basemat,\curr[\qmat])$ is constant but $\sup_{\coord,\run} \qval_{\coord,\run}=\infty$.
Then, letting $\curr[\qval]^{+} = \max_{\coord} \qval_{\coord,\run}$ and $\curr[\qval]^{-} = \min_{\coord} \qval_{\coord,\run}$, we will also have $\limsup_{\run\to\infty} (\curr[\qval]^{+} - \curr[\qval]^{-}) = \infty$ by the fact that $\curr[\qmat] \in \qmats$ (so $\tr\curr[\qmat] = 0$). 
Then, by relabeling indices and descending to a finer subsequence if necessary, we may assume without loss of generality that the following properties hold concurrently:
\begin{enumerate}
\item
$\curr[\denmat] = \mirror(\curr[\qmat])$ converges to some $\denmat \in \denmats$ (by the compactness of $\denmats$).
\item
All elements of the eigen-decomposition $\curr[\denmat] = \sum_{\coord} \denval_{\coord,\run} \denvec_{\coord,\run} \denvec_{\coord,\run}^{\dag}$ of $\curr[\denmat]$ also converge (by the compactness of the unitary group $U(\vdim)$ and the unit simplex).
\item
$\qval_{\vdim,\run} = \curr[\qval]^{+}$ and $\qval_{1,\run} = \curr[\qval]^{-}$ for all $\run=\running$
\item
The index set $\{1,\dotsc,\vdim\}$ can be partitioned into two nonempty sets $\set^{+}$ and $\set^{-}$ such that
\begin{enumerate*}
[\itshape a\upshape)]
\item
$\sup_{\run} \abs{\curr[\qval]^{+} - \qval_{\coord,\run}} < \infty$ for all $\coord\in\set^{+}$;
and
\item
$\curr[\qval]^{+} - \qval_{\coord,\run} \to \infty$ for all $\coord\in\set^{-}$.
\end{enumerate*}
\end{enumerate}
Since $\curr[\denmat]$ and $\curr[\qmat]$ commute (by \cref{lem:spectraplex}), we get:
\begin{align}
\trof{\curr[\qmat] \cdot (\basemat - \curr[\denmat])}
	&= \sum_{\coord=1}^{\vdim} \qval_{\coord,\run} (\base_{\coord,\run} - \denval_{\coord,\run})
	\notag\\
	&= \sum_{\coord=1}^{\vdim} (\qval_{\coord,\run} - \curr[\qval]^{+}) (\base_{\coord,\run} - \denval_{\coord,\run})
	\notag\\
	&= \sum_{\coord\in\set^{+}} (\qval_{\coord,\run} - \curr[\qval]^{+}) (\base_{\coord,\run} - \denval_{\coord,\run})
	+ \sum_{\coord\in\set^{-}} (\qval_{\coord,\run} - \curr[\qval]^{+}) (\base_{\coord,\run} - \denval_{\coord,\run})
\end{align}
where we let $\base_{\coord,\run} = \denvec_{\coord,\run}^{\diag} \basemat \denvec_{\coord,\run}$, and we used the fact that $\tr\basemat = \tr\curr[\denmat] = 1$ in the second line.
Now, by the third part of \cref{lem:spectraplex}, we have $\denval_{\coord,\run} \to 0$ for all $\coord\in\set^{-}$, so $\lim_{\run\to\infty} (\base_{\coord,\run} - \denval_{\coord,\run}) > 0$ for all $\coord\in\set^{-}$ (recall that $\base_{\coord,\run}$ converges and $\basemat$ has been assumed full-rank).
As a result, we get
\begin{equation}
\sum_{\coord\in\set^{-}} (\qval_{\coord,\run} - \curr[\qval]^{+}) (\base_{\coord,\run} - \denval_{\coord,\run})
	\to -\infty
	\quad
	\text{as $\run\to\infty$}.
\end{equation}
and hence, since all quantities involved in the sum over $\coord\in\set^{+}$ are bounded, we conclude that $\trof{\curr[\qmat] \cdot (\basemat - \curr[\denmat])} \to -\infty$ as $\run\to\infty$.
However, by the definition of the Fenchel coupling, we have
\begin{equation}
\fench(\basemat,\curr[\qmat])
	= \hreg(\basemat) + \hconj(\curr[\qmat]) - \trof{\curr[\qmat] \basemat}
	= \hreg(\basemat) - \hreg(\curr[\denmat]) - \trof{\curr[\qmat] \cdot (\basemat - \curr[\denmat])},
\end{equation}
which in turn yields $\lim_{\run\to\infty} \fench(\basemat,\curr[\qmat]) = -\infty$, a contradiction.
\end{proof}

\para{Step 5: Putting everything together}

With all this said and done, we are finally in a position to prove \cref{thm:recurrence}.

\begin{proof}[Proof of \cref{thm:recurrence}]
By \cref{lem:bounded} and the fact that $\fench(\eq,\qmat(\time))$ is a constant of motion under \eqref{eq:FTQL-0}, it follows that the trajectories of \eqref{eq:FTQL-0} remain bounded.
Hence, by \cref{prop:volume-0} and Poincaré's recurrence theorem, it follows that the dynamics \eqref{eq:FTQL-0} are recurrent, \ie for (Lebesgue) almost every initial condition $\qmat \in \qmats$, the induced trajector $\qmat(\time) \equiv \qflow_{\time}(\qmat)$ returns arbitrarily close to $\qmat$ infinitely many times.
More precisely, this means that for every neighborhood $\qnhd$ of $\qmat(\tstart)$ in $\qmats$, there exists some finite return time $\hit_{\qnhd}$ such that $\qflow_{\hit_{\nhd}}(\qmat) \in \qnhd$.
Hence, by taking a shrinking net of balls $\qnhd_{\run}$ centered at $\qmat$ and with radius $1/\run$, and letting $\time_{\run} = \hit_{\qnhd_{\run}}$, we deduce that $\qflow(\time_{\run}) \to \qmat$ as $\run\to\infty$.
Our assertion then follows by applying the above construction to some $\qmat \in \mirror^{-1}(\stateof{\tstart}) \cap \qmat$ and noting that $\stateof{\time_{\run}} = \pflow_{\time_{\run}}(\stateof{\tstart}) = \mirror(\qflow_{\time_{\run}}(\qmat)) \to \mirror(\qmat) = \stateof{\tstart}$.
\end{proof}

\section{Numerical Results}
\label{app:numerics}

In this appendix, we describe our simulations setup. We consider a simple two-player zero-sum game with payoff matrices $P_1 = 
\begin{pmatrix}
2 & 1\\
-2 & -1
\end{pmatrix}$ for the row-
and $P_2 = -P_1$ for the column-player, expressed in the quantum regime. Denoting the rows as $(\pure_1,\pure_2)$ and the columns as $(\purealt_1,\purealt_2)$, we have that the action profile $(\pure_1,\purealt_2)$ is a strict $\Nash$. We run \eqref{eq:MMW} on the two-player quantum game with initial conditions $(0.2,0.8)$ and $(0.8,0.2)$ for the row- and column-player, respectively, and generate the trajectories of the players' states up to time $\time = 100$. The way we visualize them is through each player's Bloch sphere. Note that the red diamond ``\textcolor{BrickRed}{$\blacklozenge$}'' on the figures denotes the initial mixed states of the players. As we observe, the trajectories converge to the boundary of the Bloch spheres, \ie to pure states. To facilitate the reproducibility of our results, we used we denote that we used the openly available code from: [\citenum{JPS22}, \href{https://openreview.net/forum?id=FFZYhY2z3j}{https://openreview.net/forum?id=FFZYhY2z3j}] 


\bibliographystyle{icml2023}
\bibliography{bibtex/IEEEabrv,bibtex/Bibliography-PM}

\end{document}